\documentclass[a4paper, 11pt]{amsart}

\def\RR{{\mathbb R}}

\usepackage{amsmath,amssymb,amsfonts,amsthm}
\usepackage[utf8]{inputenc}
\usepackage[autostyle=true]{csquotes}
\usepackage{xspace}
\usepackage{textcomp} 
\usepackage{tikz}
\usepackage{caption}
\usepackage{subcaption}
\usetikzlibrary{calc,shapes} 
\usepackage{tikz-cd}
\usepackage{stmaryrd}
\usepackage{graphicx}
\usepackage{faktor}
\usepackage{enumitem}
\usepackage{booktabs}
\usepackage{algorithm2e} 
\usepackage{pgfplots}
\usepgfplotslibrary{external} 
\pgfplotsset{compat=1.14}
\usepackage{lineno}

\newcommand\cS{\mathcal{S}}

\newcommand\expectation{\mathbb E}
\DeclareMathOperator\cov{cov}

\newcommand\normalDistribution{\mathcal{N}}
\newcommand\foldedNormalDistribution{\mathcal{FN}}
\newcommand\bigOh{\mathcal{O}}

\newcommand\fec{{\rm fec}\xspace} 
\newcommand\ttd{{\rm ttd}\xspace} 
\newcommand\dev{{\rm dev}\xspace} 

\newcommand\subdivision{\cS}
\newcommand\dualgraph{\Gamma}
\newcommand\significantsubgraph{\dualgraph_{\rm sig}}
\newcommand\significantclusterpartition{\subdivision_{\rm sig}}

\newcommand\SetOf[2]{\left\{#1\,\vphantom{#2}\right|\left.\vphantom{#1}\,#2\right\}}

\DeclareMathOperator{\conv}{conv}
\DeclareMathOperator{\nvol}{nvol}
\DeclareMathOperator{\LP}{LP}

\theoremstyle{plain}
\newtheorem{theorem}{Theorem}
\newtheorem{proposition}[theorem]{Proposition}
\newtheorem{lemma}[theorem]{Lemma}

\theoremstyle{definition} 

\newtheorem{example}[theorem]{Example}
\newtheorem{remark}[theorem]{Remark} 

\newcommand\polymake{{\tt polymake}\xspace}
\usepackage{comment}

\usepackage[colorinlistoftodos,bordercolor=orange,backgroundcolor=orange!20,linecolor=orange,textsize=scriptsize]{todonotes}

\makeatletter
\renewcommand{\todo}[2][]{\tikzexternaldisable\@todo[#1]{#2}\tikzexternalenable}
\makeatother

\title[Cluster partitions and fitness landscapes]{Cluster partitions and fitness landscapes of the Drosophila fly microbiome}

\author[H.\,Eble]{Holger Eble}
\address{\textsuperscript{1}
  Institut f\"ur Mathematik, MA 6-2, TU Berlin, 10623 Berlin, Germany}
\email{eble@math.tu-berlin.de}

\author[M.\,Joswig]{Michael Joswig}
\address{\textsuperscript{1}
    Institut f\"ur Mathematik, MA 6-2, TU Berlin, 10623 Berlin, Germany}
\email{joswig@math.tu-berlin.de}
\thanks{Research by M. Joswig is partially supported by Einstein Stiftung Berlin and Deutsche Forschungsgemeinschaft (EXC 2046: \enquote{MATH$^+$}, SFB-TRR 109: \enquote{Discretization in Geometry and Dynamics}, SFB-TRR 195: \enquote{Symbolic Tools in Mathematics and their Application}, and GRK 2434: \enquote{Facets of Complexity}).}

\author[L.\,Lamberti]{Lisa Lamberti}
\address{\textsuperscript{2} Department of Biosystems Science and Engineering, ETH Z\"urich, Basel, Switzerland;
SIB Swiss Institute of Bioinformatics, Basel, Switzerland}
\email{lisa.lamberti@bsse.ethz.ch}

\author[W.\,Ludington]{William B. Ludington}
\address{\textsuperscript{3}
Department of Embryology, Carnegie Institution for Science, Baltimore, USA}
\email{ludington@carnegiescience.edu}
\thanks{W.B. Ludington acknowledges support by the NIH Office of the Director's Early Independence Award grant DP5OD017851}

\begin{document}

\tikzset{
  green3.845/.style = {fill=green!20!brown!40},
  green3.075/.style = {fill=green!40!brown!30},
  green2.245/.style = {fill=green!60!brown!20},
  green0.065/.style = {fill=green!80!brown!10},    
}

\begin{abstract}
  The concept of genetic epistasis defines an interaction between two genetic loci as the degree of non-additivity in their phenotypes. A fitness landscape describes the phenotypes over many genetic loci, and the shape of this landscape can be used to predict evolutionary trajectories. Epistasis in a fitness landscape makes prediction of evolutionary trajectories more complex because the interactions between loci can produce local fitness peaks or troughs, which changes the likelihood of different paths.  While various mathematical frameworks have been proposed to calculate the shapes of fitness landscapes, Beerenwinkel, Pachter and Sturmfels \ (2007) suggested studying regular subdivisions of convex polytopes. In this sense, each locus provides one dimension, so that the genotypes form a cube with the number of dimensions equal to the number of genetic loci considered. The fitness landscape is a height function on the coordinates of the cube. Here, we propose cluster partitions and cluster filtrations of fitness landscapes as a new mathematical tool, which provides a concise combinatorial way of processing metric information from epistatic interactions. Furthermore, we extend the calculation of genetic interactions to consider interactions between microbial taxa in the gut microbiome of \emph{Drosophila} fruit flies. We demonstrate similarities with and differences 
to the previous approach. As one outcome we locate interesting epistatic information on the fitness landscape where the previous approach is less conclusive.
\end{abstract}

\maketitle

\tableofcontents




\section{Introduction}

In evolutionary biology, the concept of a fitness landscape plays a prominent role in the study of genetic mutations, evolutionary trajectories, and the consequences for organismal health and disease ~\cite{Krug}.
These landscapes were introduced by Sewall Wright in \cite{Wright1932} and typically arise as high dimensional discrete or continuous genotype--phenotype mappings.
The underlying coordinates in these mappings encode alleles at $n$ genetic loci of interest and are called \emph{genotypes}.
The convex hull of these genotypes, which might be viewed as points in $\RR^n$, is a polytope called the \emph{genotope}. In non-technical terms, we can think of each genetic locus as a separate dimension. The collection all combinations of these loci (e.g. of double mutants, triple mutants, etc.) forms a cube where the dimension $n$ is the number of loci considered. Thus, for a two locus, bi-allelic set, the genotope is a square (i.e., two-dimensional cube) with vertices (i.e., corners) representing the wild type, each single mutant, and the double mutant.

Fitness landscapes then arise by mapping the reproductive success (or other fitness traits) of each genotype to its corresponding vertex on the cube. The concept of genetic epistasis defines an interaction between two genetic loci as the degree of non-additivity in their phenotypes. 
Shapes of fitness landscapes reveal the epistasis and are regular subdivisions of the genotope induced by genotype--phenotype mappings, meaning that the degree to which a set of vertices interact has a physical shape that we can measure (as studied in Beerenwinkel et al.~\cite{BPS:2007,BMC:2007}). 
Such subdivisions play a key role in determining interaction patterns among altered genes and pathways on the genotope, and they shed light on the possible orders in which genetic mutations might occur. For instance, where an epistatic interaction lowers fitness, evolutionary paths traversing that vertex would have lower odds of occurring.
The study of interaction patterns is a general one that applies also to economies, social networks, and food webs in ecology.

Recently the gut microbiome has arisen in biology as a major factor shaping the genotype--phenotype mappings in animals \cite{Kreznar:2017gr}.
The microbiome itself is an ecological interaction network of microbial species.
Resolving the structure of the microbiome interaction patterns and their impacts on genotype--phenotype mappings in the host is a major unsolved problem in biology.
However, the number of possible interactions is massive in general, as the number of species in the microbiome is on the order of hundreds to thousands. Thus, there is a great need to develop methods capable of detecting interactions without drowning in data.
Here, we focus on a naturally simple gut microbiome, the \emph{Drosophila} fruit fly\textquotesingle s gut, with only five species of bacteria.
For the purposes of the present paper, we keep the notations and terminology of genetic interactions, noting that they apply also to microbiome interactions.

In this work, we build on the approach developed by Beerenwinkel et al.~\cite{BPS:2007,BMC:2007} and propose a new method to process metric interaction information, also known as epistasis, usually arising from interactions among altered genes.
The main idea we bring to the theory of interactions are cluster partitions and cluster filtrations.
These deal with connected components, called \emph{clusters}, of some subgraph of the dual graph associated to a regular subdivision that is induced, e.g., by a genotype--phenotype mapping.
To build these clusters, we first associate positive weights to pairs of maximal and adjacent simplices in the regular subdivision.
Similar to an angle, these weights measure the deviation of the two adjacent simplices from being affinely dependent. 
Using these weights we form progressively bigger clusters by gluing the ones with lower weight together in a continuous process, which we call \emph{cluster filtration}.
The important new aspect in this algorithm is that the filtration process is designed to statistically distinguish an essential biological signal, encoded in the positive weight, from noise.
In this way, filtrations enable us to handle and interpret the usually vast interaction information, which is currently only
fully characterized for double and triple mutants; cf.\ \cite{BPS:2007,BMC:2007}.

To describe the strength of our approach, we consider fitness landscapes for microbiome-modified \emph{Drosophila} flies. We then describe similarities and differences to previous approaches.
More precisely, the data set we inspect in this work consists of \emph{Drosophila} flies prepared with up to five different bacterial species in their gut.
We then view each bacterial combination as a genotype. In this way, the genotope is given by a 5-dimensional cube.
The fitness landscapes we consider are defined by daily fecundity (referred to as \fec), time to death (\ttd) and development time (\dev) mappings.
For each such fitness landscape we study the induced regular subdivision and describe their properties in the language of clusters and cluster filtrations.
To compare the epistatic information between fitness landscapes we use cluster partitions.
Our results show that cluster filtrations detect interactions when these are present in the sense of \cite{BPS:2007,BMC:2007}. Additionally, we locate statistically relevant and previously undetected epistatic information.
Comparing our findings with previous studies confirms that cluster filtrations can also be used to strengthen existing analysis
and prompt new possible conclusions in interaction networks.
Cluster partitions and their filtrations are a new mathematical idea, and we hope that this will find applications also beyond \emph{Drosophila} microbiome data.




\section{Mathematical background and terminology}
\label{sec:background}

Our approach relies on the theoretical framework for revealing epistatic interactions 
in genetic systems given in \cite{BPS:2007,BMC:2007}, and we use the same terminology. 
The theory of regular subdivisions is developed in the 2010 monograph by DeLoera, Rambau, and Santos~\cite{Triangulations}.

\subsection{Genotopes and their regular subdivisions}
We consider a fixed $n$-dimensional convex polytope $P$ in $\RR^n$.
That is, $P$ is the convex hull of finitely many points and we will assume that $P$ affinely spans the entire space.
A point $v\in P$ for which there exists an affine hyperplane which meets $P$ only in $v$ is called a \emph{vertex} of $P$.
The set of vertices, denoted as $V$, forms the unique minimal set which generates $P$ as the convex hull.
Our second ingredient is a \emph{height function} on the vertices, which is any function $h:V\to\RR$ that assigns a real number to each vertex of $P$.

We will be particularly interested in the case where $V=\{0,1\}^{n}$, and $P=[0,1]^n$ is the $n$-dimensional unit cube.
Following the approach of \cite{BPS:2007,BMC:2007} we call $[0,1]^{n}$ the \emph{genotope} of an $n$-biallelic system. 
The vertices in $\{0,1\}^{n}$ are identified with binary strings of length $n$ called \emph{genotypes}.
In the biological applications we have in mind, the points in the genotope correspond to the allele frequencies in a population; cf.\ Section~\ref{subsec:allele-freqs} below.
Height functions then correspond to traits, such as reproductive fitness of an organism or other experimental measurements ---also called \emph{phenotypes}--- on the genotypes of the $n$-biallelic system.

The set of lifted points
\[
  V(h) \ := \ \SetOf{(v,h(v))}{v\in V}
\]
generates a polytope $P(h):=\conv V(h)$ in $\RR^{n+1}$.
We will assume that the height function $h$ is \emph{nontrivial} in the sense that the lifted polytope $P(h)$ has full dimension $n+1$.
By construction, the points in $V(h)$ are precisely the vertices of $P(h)$.
In general, there are three types of facets of $P(h)$: if $\nu$ is an outward normal vector on the facet $F$, then $F$ is called an \emph{upper/vertical/lower facet} of $P(h)$ if the $(n{+}1)$-st coordinate of $\nu$ is positive/zero/negative.
It may happen that there are no vertical facets, but there are always upper and lower facets.
The upper facets form a polyhedral ball sitting in the boundary complex of the lifted polytope $P(h)$.
Projecting them back to $P$, by omitting the last coordinate, yields a polyhedral subdivision $\subdivision=\subdivision(V,h)$ of $P$.
A \emph{polyhedral subdivision} of $P$ is a finite family of polyhedra, whose elements we call the \emph{cells} of the subdivision, such that each face of a cell is a cell, and the intersection of any two cells is a (possibly empty) cell.
Those subdivisions which are induced by a height function are called \emph{regular}; cf.\ Definition 2.3.1 and Lemma 2.3.11 in \cite{Triangulations}. 
A polyhedral subdivision for which all cells are simplices is called a \emph{triangulation}.
Triangulations induced by a height function 
are generic in the following sense.
If each value of the height function is chosen at random (e.g., uniformly in a fixed interval) then the induced regular subdivision is a triangulation almost surely.

The height function $h$ is called a \emph{fitness landscape} in \cite[\S3]{BPS:2007}.
The genotope subdivision $\subdivision(V,h)$ is known as the \emph{shape} of the fitness landscape $h$.
Geometric properties of these shapes of fitness reflect interactions among organisms or genotypes.
In the biological setting we have in mind, height functions are given by certain fittings of replicated measurements. 
Depending on the nature of the data, these fittings can, for instance, be means, medians or modes of the observed measurements, as well as expected values of nonparametric density estimators, described as in \cite[Chapter 9]{MR1948588}.
As such we can assume that the height functions are generic, so they 
induce regular triangulations of the $n$-cube $[0,1]^{n}$.
In the concrete computations below, we consider means of replicated measurements.

Before we continue with our exposition, let us consider an example which we will revisit later.
This is about the smallest nontrivial case which arises.
\begin{example}\label{ex:running_ex}
  We consider a $2$-biallelic system, and so the genotope is the unit square $P=[0,1]^2$.
  Its four vertices $00$, $10$, $01$ and $11$ form the set~$V$; here, $01$ is shorthand notation for the point with coordinates $(0,1)$.
  Assume that some measurement gives the height function $h$ which reads
  \begin{equation}\label{eq:running_ex}
    \begin{array}{llll}
      h(00) = 53.25 & h(10) = 46.65 & h(01) = 43.16 &  h(11) = 43.48
    \end{array}
    \enspace .
  \end{equation}
  The lifted polytope $P(h)$ is a $3$-dimensional simplex (a.k.a.\ tetrahedron).
  In this case there are two upper and two lower facets; no vertical ones.
  Figure~\ref{fig:running_ex} shows the upper facets of $P(h)$ and the resulting genotope subdivision $\subdivision(V,h)$.
  The latter is a triangulation with two maximal cells, indicated in green and red.
\end{example}

\begin{figure}[ht]
  \begingroup
  \newcommand\radius{2.5}
    \includegraphics[width=\textwidth]{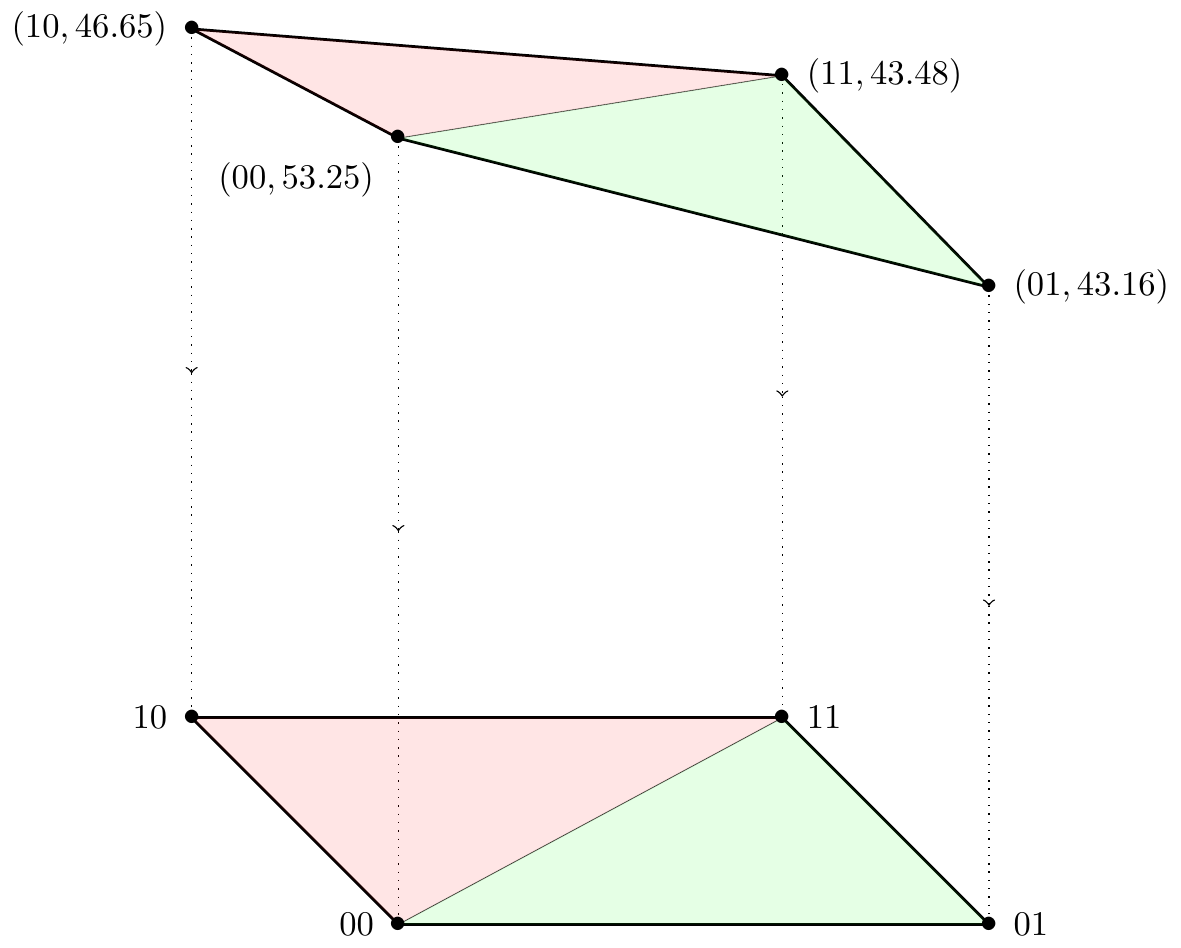}  
  \endgroup
  \caption{%
    Upper facets of the lifted polytope $[0,1]^2(h)$ and the induced regular triangulation $\subdivision(\{0,1\}^2,h)$.}
  \label{fig:running_ex}
\end{figure}

\begin{remark}
  Most polytopes admit triangulations which are not regular, i.e., not induced by any height function; cf.\ \cite[Theorem~6.3.11]{Triangulations} for examples of non-regular triangulations of $[0,1]^n$ for $n\geq 4$. While the triangulations of interest here will always be regular, the algorithmic methods discussed below generalize to the nonregular setting.
\end{remark}

\subsection{Fittest populations}
\label{subsec:allele-freqs}
Again we consider the genotypes $V=\{0,1\}^n$ of an $n$-biallelic system.
A map $p:V\to\RR$ is a \emph{(relative) population} if it attains only nonnegative values which sum to one.
This yields a point
\[
  \rho \ = \ \rho(p) \ := \ \sum_{v\in V} p(v)v
\]
in the genotope $[0,1]^n$, which is the \emph{allele frequency vector}.
Its coordinate $\rho_i$ describes the probability for the population $p$ to have allele $1$ in its $i$-th locus.
The set of all relative populations, denoted as $\Delta_V$, is a simplex of dimension $2^n-1$.

Now we add the height function $h:V\to\RR$ to the picture.
For a fixed allele frequency vector $w\in[0,1]^n$ this gives rise to the linear program
\begin{equation}\label{eq:LP}
  \begin{array}{ll}
    \text{maximize} & h \cdot p\\
    \text{subject to} & p\in\Delta_V \text{ and } \rho(p)=w \enspace.
  \end{array}
  \tag{$\LP(h,w)$}
\end{equation}
The coordinates of the vector $p$ are the variables to be determined.
If $h$ and $w$ both are generic then $\LP(h,w)$ has a unique optimal solution, the \emph{fittest population} $p^*=p^*(h,w)$, and this a vertex of the polytope
\[
  \Delta_{V,w} \ := \ \SetOf{p\in\Delta_V}{\rho(p)=w} \enspace ,
\]
which is contained in the $2^n$-dimensional vector space $\RR^V$.
The condition $\rho(p)=w$ is equivalent to $n$ linear equations, one for each coordinate of the allele frequency vector $w$.
It follows that the fittest population $p^*$ is the convex combination of at most $n+1$ vertices of $\Delta_V$, and the projection $\rho(p^*)=w$ gives rise to a representation of $w=\lambda_1v_1+\dots+\lambda_{n+1}v_{n+1}$, with $\lambda_i\geq 0$ and $\sum\lambda_i=1$, as the convex combination of at most $n+1$ genotypes $v_{1},\dots,v_{n+1}\in V$.
These genotypes are precisely the vertices of the unique simplex $s$ of $\subdivision(V,h)$ which contains $w$.
The genericity of $h$ implies that $\subdivision(V,h)$ is a triangulation, 
while the genericity of $w$ implies that the simplex $s$ is unique.
The optimal value of $\LP(h,w)$ is $h \cdot p^*=\sum \lambda_i(h(v_{i}))$.
In this way we obtain the piecewise linear function
\[
  \begin{split}
    h^* \colon [0,1]^{n}&\longrightarrow \RR\\
    w &\longmapsto h \cdot p^*(h,w)
  \end{split}
\]
on the genotope.
Now the regions of linearity of $h^*$ coincide with the maximal cells 
of the regular triangulation $\subdivision(V,h)$. 

Applying our methods to measurement data will almost always establish generic height functions, and thus the relevant polyhedral subdivisions will almost always be triangulations.

\begin{remark}
For every fitness landscape $h$ there are two shapes. 
One shape is induced by the upper facets of the convex hull 
$P(h)$, the other by the lower facets.
Considering simultaneously both shapes might be advisable.
However, previous approaches \cite{BPS:2007,BMC:2007} adopted the convention of considering upper facets in the definition of a regular subdivision, see also \cite{Hallgrimsdottir:2008aa}.
This convention is consistent with modelling the fitness of a population via the linear program $\LP(h,w)$, which is a maximization problem.

With a few straightforward technical adjustments, such as phrasing fitness in terms of a minimization problem, all our results also hold using lower facets.
\end{remark}

\subsection{The dual graph of a subdivision and its complexity}
Let $\subdivision$ be a polyhedral subdivision of some $n$-dimensional polytope $P$.
Two maximal cells of $\subdivision$ are \emph{adjacent} if they share a common $(n-1)$-dimensional cell. 
This adjacency relation induces a graph structure as follows: the nodes are the maximal cells (of dimension $n$), and an edge connects two nodes if the two cells are  adjacent.
This is known as the \emph{dual graph} of $\subdivision$, and we denote it by $\dualgraph(\subdivision)$.
Notice that $\Gamma(\subdivision)$ is always connected.
The edges of $\dualgraph(\subdivision)$ are the \emph{dual edges} of~$\subdivision$.

The lemma below gives essential complexity bounds in the case of most interest to us.
Let us denote the minimal number of maximal cells of any triangulation of $[0,1]^n$ by $k_*(n)$.

\begin{lemma}\label{lem:complexity}
  Let $\subdivision$ be a triangulation of the unit cube $[0,1]^n$, and let $k$ be the number of nodes of $\dualgraph(\subdivision)$.
  Then
  \begin{equation}\label{eq:complexity}
       2^n-n \ \leq \ k_*(n) \ \leq \ k \ \leq n! \enspace .
  \end{equation}
  The lower bound is attained if and only if $\Gamma(\subdivision)$ has no cycles.
  Moreover, the number of dual edges is at most $k(n+1)/2 - n\cdot k_*(n-1) < (n+1)!$.
  \end{lemma}

\begin{proof} The first part of the claim is a special case of \cite[Theorem 2.6.1]{Triangulations}.
 For the second part, observe that each $n$-simplex is adjacent to at most $n+1$ other simplices as this is the number of its facets.
  Yet $\subdivision$ induces a triangulation on each of the $2n$ facets of $[0,1]^{n}$.
  Therefore, at least $2n \cdot k_*(n-1)$ of the $k(n+1)$ cells of 
  dimension $n-1$ in $\subdivision$ lie in the boundary of $[0,1]^{n}$.
  These $(n{-}1)$-cells are contained in a unique maximal one.
  We arrive at the estimate of at most
  \[
    k(n+1) - 2n \cdot k_*(n-1) \ \leq \ (n+1)! - 2n(2^{n-1}-n+1)
  \]
  incident pairs of nodes and edges of $\Gamma(\subdivision)$.
  Dividing by two gives the upper bound on the number of dual edges.
\end{proof}

For $n=3$ we get $2^3-3=5$ as the lower bound in \eqref{eq:complexity}, and $3!=6$ as the upper bound.
Both bounds are tight; i.e., there are triangulations of $[0,1]^3$ with five and six facets, respectively.
While, for $n=4$, the lower bound in \eqref{eq:complexity} is only $2^4-4=12$ we have $k_*(4)=16$; cf.\ \cite[Example 6.3.14]{Triangulations}.
To determine the exact lower bound $k_*(n)$ for $n\geq 8$ is a difficult open problem; cf.\ \cite[\S6.3.3]{Triangulations} and Table~\ref{tab:cube-complexity} for an overview.
\begin{table}[ht]
  \caption{The minimal number $k_*(n)$ of maximal cells of a triangulation of $[0,1]^n$.}
  \label{tab:cube-complexity}
  \begin{tabular*}{.97\linewidth}{@{\extracolsep{\fill}}ccccccccc@{}}
    \toprule
    $n$ & 1 & 2 & 3 & 4 & 5 & 6 & 7 & 8 \\
    \midrule
    $k_*(n)$ & 1 & 2 & 5 & 16 & 67 & 308 & 1493 & $\leq 11\, 944$ \\[1.2ex]
    $2^n-n$ & 1 & 2 & 5 & 12 & 27 & 58 & 121 & 248 \\[1.2ex]
    $ \frac{2^nn!}{(n+1)^{(n+1)/2}}$ & 1 & 1.54 & 3 & 6.87 & 17.78 &50.78 & 157.5 & 524.41 \\
    \bottomrule
  \end{tabular*}
\end{table}

The proof of Lemma~\ref{lem:complexity} is based on the interplay between the size of a triangulation and the volumes of its maximal cells.
This can be carried further to derive bounds which are better asymptotically.
The key ingredient is Hadamard's famous problem of giving an upper bound for the determinant of a matrix with given entries; cf.\ \cite{Hadamard+problem} for a survey.
In our context this yields the following.
\begin{lemma}\label{lem:hadamard}
  The normalized volume of any simplex spanned by vertices of $[0,1]^n$ is bounded by
  \[
    \frac{(n+1)^{(n+1)/2}}{2^n} \enspace.
  \]
\end{lemma}
This can be employed to derive a lower bound for $k_*(n)$.
The even better bound
\begin{equation}\label{eq:hadamard}
  \frac{2^nn!}{(n+1)^{(n+1)/2}} \ \leq \ k_*(n)
\end{equation}
arises from the Hadamard inequality for matrices with $\pm 1$ coefficients; cf.\ \cite[\S6.3.3]{Triangulations}.
Starting from $n\geq7$ the bound \eqref{eq:hadamard} is better than the more naive lower bound $2^n-n$ from Lemma~\ref{lem:complexity}.




\section{Clusters in fitness landscapes}\label{sec:clusters}

Our next goal is to show how epistatic information can be extracted from the dual graph associated to the triangulation of the genotope $[0,1]^{n}$ by a given height function.
To do this we first associate a positive weight with each dual edge, we then relate this information to interaction coordinates and epistasis.
Later, we introduce cluster and cluster filtrations as a new tool to filter, summarize and analyze epistatic information. 




\subsection{Epistatic weight}\label{subsec:epistaticweight}

As before, let $P$ be an arbitrary $n$-polytope in $\RR^{n}$, equipped with a generic height function $h$.
This induces a regular triangulation $\subdivision=\subdivision(V,h)$, where $V$ is the vertex set of $P$.
Let $s$ and $t$ be two adjacent $n$-simplices in $\subdivision$.
Then there are altogether $n+2$ vertices $v_1,v_2,\dots ,v_{n+2}$ of $P$ such that
\[
  s \ = \ \conv\{v_1,v_2,\dots,v_{n+1}\} \quad \text{and} \quad t \ = \ \conv\{v_2,v_3,\dots,v_{n+2}\} \,.
\]
We consider the $(n+2){\times}(n+2)$-matrix
\begin{equation}\label{eq:epistatic-matrix}
  E_h(s,t) \ := \
  \begin{pmatrix}
    1 & v_{1,1} & v_{1,2} & \dots & v_{1,n} & h(v_1) \\
    1 & v_{2,1} & v_{2,2} &  \dots & v_{2,n} & h(v_2) \\
    \vdots  & \vdots &  \vdots & \vdots & \vdots &\vdots \\
    1 & v_{n+2,1} & v_{n+2,2} & \dots & v_{n+2,n}  & h(v_{n+2})
  \end{pmatrix} \enspace ,
\end{equation}
where $v_{i1},v_{i2},\dots,v_{in}$ are the coordinates of $v_i\in\RR^n$.
The \emph{epistatic weight} of the dual edge connecting $s$ and $t$ is then defined as 
\begin{equation}\label{eq:epistatic-weight}
  e_h(s,t) \ := \ \bigl|\det E_h(s,t)\bigr| \cdot \frac{\nvol(s\cap t)}{\nvol{s}\cdot\nvol{t}} \enspace ,
\end{equation}
where $\nvol{s}$ is the \emph{normalized volume} of $s$, i.e., the determinant of the submatrix $N(s)$ 
obtained from $E_h(s,t)$ by omitting the last column and the row corresponding to the vertex $v_{n+2}$ not lying in $s$.
Similarly, $N(t)$ is the submatrix of $E_h(s,t)$ obtained by omitting the last column and the row corresponding to the vertex $v_1$ not lying in $t$.
The determinant of $E_h(s,t)$ is the volume of the convex hull of the $(n+2)$ vertices of $s$ and $t$ lifted to $\RR^{n+1}$ by the height function $h$.
The intersection $s \cap t$ is spanned by $v_2,v_3,\dots,v_{n+1}$ and separates the two \emph{satellite vertices} $v_1$ and $v_{n+2}$. Its normalized $(n-1)$-dimensional volume $\nvol(s\cap t)$ coincides with the $n$-dimensional normalized volume of a pyramid over $s\cap t$ with height $1$. 

The epistatic weight $e_h(s,t)$ vanishes if the lifted point configuration of $s \cup t$ with respect to $h$ lies in a hyperplane and is positive otherwise. If the denominator of \eqref{eq:epistatic-weight} is one, we say that
the simplices $s$ and $t$ are {\em unimodular}. One then says that a {\em triangulation is unimodular}, if all its simplices are unimodular.

For $\lambda>0$, the scaled height function $\lambda h$ over the scaled genotope $\lambda P$ provides the same combinatorial data, i.e., the same labeled maximal cells.
The factor ${\nvol(s\cap t)}/{(\nvol{s}\cdot\nvol{t})}$ in \eqref{eq:epistatic-weight} makes the epistatic weight $e_{\lambda h}(s,t) =e_{h}(s,t)$ invariant under the scaling by $\lambda$.

If $h'$ is non-trivial and not generic, then $\subdivision'=\subdivision(V,h')$ is a non-trivial regular subdivision. By Lemma 2.3.4 in \cite{Triangulations}, all maximal cells of $\subdivision'$ are full-dimensional, but not necessarily simplices. Let $C$ and $D$ be two maximal cells in $\subdivision'$ which are adjacent, i.e., the intersection  $C\cap D$ is a common face of codimension one. 
A maximal collection of points $v_1,v_2,\dots,v_{n+2}$ is called a \emph{bipyramid} in the dual edge $(C,D)$ of $\subdivision'$ if $v_1,\dots,v_{n+1}$ are affinely independent vertices of $C$ and $v_2,\dots,v_{n+2}$ are affinely independent vertices of $D$. In this way a bipyramid spans a pair $(s,t)$ of adjacent $n$-simplices contained in the union of $C$ and $D$, and $s\cap t\subseteq C\cap D$. For each such pair $(s,t)$ one can form the matrix $E_{h'}(s,t)$ as in \eqref{eq:epistatic-matrix}, and find the corresponding epistatic value $e_{h'}(s,t)$ via \eqref{eq:epistatic-weight}. To define  $e_{h'}(C,D)$, we take the mean over all $e_{h'}(s,t)$ where $(s,t)$ runs through the bipyramids in the dual edge $(C,D)$.

\begin{example}\label{ex:running_part2}
  We continue Example~\ref{ex:running_ex} where $P=[0,1]^2$, and $h$ is given by \eqref{eq:running_ex}.
  The induced triangulation has precisely two maximal cells. 
  The dual graph has a single edge connecting these two cells.
  We compute the epistatic weight
  \begin{equation}\label{eq:running_epi}
    \begin{split}
      e_{h}(\{00,10,11\},\{00,01,11\}) \ &= \ \det\begin{pmatrix}
        1 & 0 & 0 & 53.25\\
        1 & 1 & 0 & 46.65\\
        1 & 0& 1 & 43.16\\ 
        1 & 1 & 1 & 43.48
      \end{pmatrix} \cdot \sqrt{2}  \ \approx \ 9.786 \enspace .
    \end{split}
  \end{equation}
  Notice that the denominator in \eqref{eq:epistatic-weight} is one.
  The factor $\sqrt{2}$ is the one-dimensional volume of the shared face $\conv(\{00,11\})$.
  For this 2-locus system, the computation \eqref{eq:running_epi} thus agrees with the usual
  epistasis formula of \cite[Example 3.7]{BPS:2007} up to the factor $\sqrt{2}$, which does not depend on~$h$:
  \[
    \epsilon(00,11,10,01) \ := \ h(00)+h(11)-h(10)-h(01) \enspace.
  \]
  Biologically, the non-vanishing of the epistatic weight means that the additive epistatic assumption is violated:
  the fitness of the double mutant is higher than what one would expect by knowing the fitness of the single mutants and the wild type. 
\end{example}

\begin{remark}
Since we did not fix orderings of the vertices of $s$ and $t$, the matrix $E_h(s,t)$ is only defined up to row reordering.
However, our approach solely rests on the epistatic weights from \eqref{eq:epistatic-weight}; taking absolute values here makes those values independent of any ordering.
\end{remark}

Let us now summarize the biological information encoded in the dual graph valued by the epistatic weight.
First, each node in $\Gamma(\subdivision)$ corresponds to an $(n+1)$-tuple of genotypes that can be realized as the support of a fittest population, i.e., of an optimal solution of $\LP(h,w)$ for some $w$.
In this sense, we can state (by a slight abuse of language) that
each edge of $\Gamma(\subdivision)$ describes 
the union of genotypes occurring in two fittest populations. 
This union consists of $n+2$ genotypes with $n$ genotypes shared by the two 
fittest populations and where the remaining two are satellites. 
The edges incident to a given node $s$ in $\Gamma(\subdivision)$ thus
encode exactly those genotypes which together with $n$ 
genotypes of $s$ form a fittest population in the above sense.

Finally, the epistatic weight associated with an edge $e=\{s,t\}$ of 
$\Gamma(\subdivision)$ measures how 
far the supporting genotypes of the two adjacent fittest populations $s$ and $t$ are away from being affinely dependent. 
In this sense, $e_h(s,t)$ can be seen as a deformation of the 
usual statistical correlation notion, see the discussion in Section \ref{sec:statsignificance1}.
The non-vanishing of the epistatic weight of $e$ thus means that knowing the 
fitness of the genotypes supported by the fittest population $s$
does not allow us to deduce the fitness of the satellite genotype not in $s$.



\subsection{Cluster partitions and epistatic filtrations}\label{subsec:clusters}

Let $\dualgraph'$ be a spanning subgraph of the dual graph $\dualgraph=\dualgraph(\subdivision)$ of the subdivision $\subdivision$.
\emph{Spanning} means that $\dualgraph'$ has all the nodes of $\dualgraph$ but some dual edges may be missing.
We call a connected component of $\dualgraph'$ a \emph{$\dualgraph'$-cluster}. 
Further, we call the partition of the nodes of $\dualgraph$ into $\dualgraph'$-clusters the \emph{$\dualgraph'$-cluster partition} of $\subdivision$.
That is, a cluster partition is an additional combinatorial structure imposed on $\subdivision$ by the choice of the spanning subgraph $\dualgraph'$.

Our next goal is to define a filtration process on $\subdivision$ by a sequence of nested cluster partitions.
For this, consider a \emph{threshold value} $\theta$, where $\theta\geq 0$.
The specification of a threshold value defines a not necessarily connected subgraph, $\dualgraph(\theta)$, of $\dualgraph$ by deleting those dual edges whose normalized 
epistatic weight exceeds $\theta$.
The $\dualgraph(\theta)$-clusters and the $\dualgraph(\theta)$-cluster partition are shortened to \emph{$\theta$-clusters} and the \emph{$\theta$-cluster partition}, respectively.

The intuition behind these concepts comes from the following.
Consider two height functions, $h'$ and $h''$, on the vertices of $P$ such that $h'$ is not generic and such that $h''$ is a small perturbation.
More precisely, we assume that there is a positive number $\epsilon$ such that $|h''(v)|<\epsilon$ for all $v\in V$.
Our assumption on $h'$ means that $\subdivision(V,h')$ is not a triangulation.
We call $h''$ \emph{sufficiently generic} for $h'$ if $h'+h''$ is generic, i.e., $\subdivision(V,h'+h'')$ is a regular triangulation.
Note that it does not suffice to require $h''$ to be generic: e.g., if $h'$ is generic, then $-h'$ is generic, too, but their sum $h'+(-h')\equiv0$ is not.
Now $\theta$-cluster partitions can detect the perturbation by $h''$ in the following sense.

\begin{theorem}\label{thm:undo}
  For every height function $h'$ there are positive numbers $\epsilon$ and~$\theta$ such that the following holds:
  If $h''$ is sufficiently generic for $h'$ and additionally satisfies $|h''(v)|<\epsilon$ for all $v\in V$, then each maximal cell of $\subdivision(V,h')$ corresponds to exactly one $\theta$-cluster of the triangulation $\subdivision(V,h'+h'')$.
\end{theorem}

\begin{proof}
  First assume that $h'$ is generic, and thus the partition $\subdivision=\subdivision(V,h')$ induced by $h'$ is a triangulation.
  Pick $\epsilon>0$ sufficiently small such that for $|h''(v)|<\epsilon$ we have $\subdivision(V,h'+h'')=\subdivision(V,h')$.
  That is, $h'$ and $h'+h''$ lie in the same secondary cone, defined as in \cite[Def. 5.2.1]{Triangulations}.
  Such an $\epsilon$ can always be found since the secondary cone of a triangulation is an open subset of $\RR^{m-n-1}$, where $m$ is the cardinality of $V$; cf.\ \cite[\S5.2.1]{Triangulations}.
  Then the claim in the generic case becomes trivial, e.g., with $\theta=0$.
  The maximal cells of $\subdivision(V,h')$ are precisely the $0$-clusters of $\subdivision(V,h'+h'')$.

  A second case arises when $\subdivision(V,h')$ has only one maximal cell given by the entire polytope $\conv(V)$.
  Then we can pick $\epsilon>0$ arbitrary and make $\theta$ sufficiently large such that the $\theta$-cluster partition of $\subdivision(V,h'+h'')$ comprises a single cluster.
  
  Now we consider the only interesting case where $h'$ is not generic and $\subdivision:=\subdivision(V,h')$ is a non-trivial regular decomposition of $V$ induced by $h'$ which is not necessarily a triangulation.
  Here we pick $\epsilon_{1}>0$ small enough such that for all sufficiently generic $h''$ with $|h''(v)|<\epsilon_1$ the subdivision $\subdivision(V,h'+h'')$ is a triangulation which refines $\subdivision$.  
  Such an $\epsilon$ always exists as shown for instance in Lemma 2.3.15 in \cite{Triangulations}. 
  This claim can equivalently be expressed by saying that $h'$ lies in the boundary of the (full-dimensional) secondary cone of $\subdivision':=\subdivision(V,h'+h'')$.

 Let $C$ and $D$ be two maximal cells in $\subdivision$ which are adjacent. Let $(s,t)$ be the
 adjacent $n$-simplices spanned by a bipyramid contained in the union of 
 $C$ and $D$, and $s\cap t\subseteq C\cap D$.
 Let $e_{h'}(s,t)$ be the corresponding epistatic value given by \eqref{eq:epistatic-weight}.
 Now we let
  \begin{equation}\label{eq:undo-theta}
    \theta \ := \ \tfrac{1}{2} \cdot \min \SetOf{ e_{h'}(s,t) }{ s\cup t \text{ bipyramid in some dual edge of }   \Gamma(\subdivision) } \enspace ,
  \end{equation}
  which is the minimum taken over a finite set of non-zero positive real numbers and thus $\theta>0$. 
  Let $(s',t')$ be adjacent $n$-simplices in the triangulation $\subdivision'$.
  We call the dual edge $(s',t')$ \emph{local} if $s'$ and $t'$ are contained in some maximal cell of $\subdivision$.
  Now, the maximal cells of $\subdivision$ belong to a $\theta$-cluster partition of $\subdivision'$ if and only if
  \begin{equation}\label{formula:cond_for_clusterind}
    e_{h'+h''}(s',t')
    \begin{cases}
      < \theta & \text{if $(s',t')$ is local dual edge}\\
      \geq \theta & \text{otherwise} \enspace .
    \end{cases}
  \end{equation}
  We observe that setting $h''\equiv0$ yields $e_{h'+h''}(s',t')=0$ for any local dual edge $(s',t')$ since then the $n+2$ vertices of $s'\cup t'$, lifted by $h'=h'+0$ are contained in some hyperplane.
  Hence, since the determinant is multilinear (and thus continuous), we can find $\epsilon_2>0$ such that all $h''$ with $|h''(v)|<\epsilon_2$ satisfy $e_{h'+h''}(s',t')<\theta$ for all local dual edges $(s',t')$. 
  An explicit expression for $\epsilon_2$ can be given in terms of the maximal minors of the matrix formed from all vertices lifted by $h$; we leave the details to the reader.
  In this way, all local dual edges of $\subdivision'$ are contained in a $\theta$-cluster. 
  For a nonlocal dual edge $(s',t')$ of $\subdivision'$ observe that $s'\cup t'$ is a bipyramid in some dual edge of $\subdivision$ and $e_{h'}(s',t')>\theta$ holds by definition \eqref{eq:undo-theta}.
  Again by continuity of $e_{h'+h''}$ in $h''$ we get an $\epsilon_{3}$-neighbourhood of $h'$ where all nonlocal dual edges of $\subdivision'$ have epistatic weight lying above $\theta$ and form singleton $\theta$-clusters.
  Setting $\epsilon:=\min(\epsilon_1,\epsilon_2, \epsilon_3)$ settles the claim.
\end{proof}

\begin{example}
  We further continue the Example~\ref{ex:running_ex}.
  Suppose $h'\equiv 0$ is the null function and $h''$ is the height function $h$ given in \eqref{eq:running_ex}.
  Then $\subdivision(\{0,1\}^2,h')$ is the trivial subdivision of $[0,1]^2$ by itself, and $\subdivision(\{0,1\}^2,h'+h'')$ is the triangulation shown in Figure~\ref{fig:running_ex}.
  The epistatic weight of the single edge is 9.786 as determined in \eqref{eq:running_epi}.
  That is, for all $\theta \geq 9.786$ the cluster partition recovers the trivial subdivision of the unit square $[0,1]^2$.
\end{example}

Varying $\theta$ yields a stepwise coarsening of $\subdivision$ into larger and larger clusters which we call the \emph{cluster filtration} of $\subdivision$.
For sufficiently small values of $\theta$ (including, e.g., $\theta=0$) the $\theta$-cluster partition simply consists of the partition of the node set of $\dualgraph$ into singletons.
On the other hand, for sufficiently large values of $\theta$ the $\theta$-cluster partition consists of a single cluster which comprises all the nodes.
Letting the threshold value vary between these two extremes provides a simple descriptor of the \enquote{biologically relevant signal} in the data and allows us to separate epistatic information from noise.
Therefore, we also use the name \emph{epistatic filtration} instead of \enquote{cluster filtration}. For an illustration of an epistatic filtration see Section \ref{subsec:extended_example}.

A threshold value $\theta$ is \emph{critical} if the cluster partition $\subdivision(\theta)$ differs from $\subdivision(\theta-\epsilon)$ for all $\epsilon>0$.
In this case, $\theta$ is necessarily the epistatic value of some dual edge, and $\subdivision(\theta)$ has strictly fewer clusters than $\subdivision(\theta-\epsilon)$.
However, the converse is not true: there may exist dual edges whose epistatic values are not critical.
An open interval $(\theta_0,\theta_1)$ with critical thresholds $\theta_0<\theta_1$ is \emph{regular} if it does not contain any critical value.

From a more geometric perspective we could also use dihedral angles instead of our epistatic weights.
This approach would yield a theoretical result similar to Theorem~\ref{thm:undo}.
Here, we refrain from doing so since in Section~\ref{sec:statsignificance1}, we will take statistical information into account.
There the height of a vertex is actually a mean value, 
and we do not see a natural way to extend the statistics to dihedral angles. 

\begin{remark}
  The \emph{tight span} of an arbitrary $n$-dimensional polyhedral subdivision $\subdivision$ (of some point configuration) is a CW-complex $\subdivision^*$ whose $0$-skeleton is formed by the maximal cells of $\subdivision$; cf.\ \cite{HerrmannJoswigSpeyer:2014}.
  The $k$-cells of $\subdivision^*$ correspond to those subsets of the maximal 
  cells of  $\subdivision^*$ which share a common cell of dimension $n-k$.
  In this way, the $1$-skeleton of $\subdivision^*$ agrees with the dual graph $\dualgraph(\subdivision)$.
  If the subdivision is regular, its tight span is a polyhedral complex.
 Assigning epistatic weights to all cells of $\subdivision^*$, not only to the edges, would open up a way to study more involved epistatic interactions.
\end{remark}

\begin{remark}
  The regular subdivision $\subdivision$ of any (rational) point configuration is dual to a tropical hypersurface \cite[\S3.1]{Tropical+Book}.
  In this way our epistatic filtrations, which are defined on the dual graph $\dualgraph(\subdivision)$, impose an additional structure on the $1$-skeleton of any tropical hypersurface.
\end{remark}

\begin{remark}
  The graph $\Gamma(\subdivision)$ equipped with the epistatic weights induces a finite metric space (on the facets of $\subdivision$), via taking shortest paths.
  Considering Vietoris--Rips filtrations then allow for studying the geometry of the lifted points by means of persistent homology \cite{Edelsbrunner+Harer:2010}.
  It could be interesting to investigate if there is any connection with the epistatic filtrations.
\end{remark}

Studying the ramifications into higher-dimensional epistatic weights, tropical geometry or persistent homology looks very promising, but all these topics are beyond the scope of this article.

\subsection{Computing epistatic filtrations}
\label{subsec:computing}

We now explain how to compute the epistatic filtration from the vertex set $V\subset\RR^n$ and the height function $h:V\mapsto\RR$ as input.
In the first step we need to determine the list of maximal cells of the regular subdivision $\subdivision=\subdivision(V,h)$; the standard encoding of each maximal cell is as a subset of $V$.
Determining $\subdivision$ is achieved by computing the convex hull of the lifted points in $\RR^{n+1}$ and selecting the facets with upward pointing normals.
Computing convex hulls is a standard problem in computational geometry \cite{HDCG:3:convex+hulls} with a somewhat delicate complexity status \cite[Open problem 26.3.4]{HDCG:3:convex+hulls};
see \cite{polymake:2017} for a recent survey from a practical point of view.

The input to the second step is the list of maximal cells of $\subdivision$ as subsets of~$V$; let $k$ denote their number.
If $V$ are the vertices of the $n$-cube then $k\leq n!$ by Lemma \ref{lem:complexity}.
The $k$ maximal cells form the nodes of the dual graph $\dualgraph(\subdivision)$.
To find the edges one can check the $\tbinom{k}{2}$ pairs of maximal cells, looking for those pairwise intersections which are maximal with respect to inclusion.
This yields the edges of $\dualgraph(\subdivision)$.
In the most relevant special case where $\subdivision$ is a triangulation, the maximal pairwise intersections are precisely those of cardinality $n$.
So the total cost for this step amounts to $\bigOh(k^2n)$.
Let $\ell$ denote the number of dual edges.
If $V$ are the vertices of the $n$-cube then $\ell\leq k(n+1)/2 - n(2^{n-1}-n+1)<(n+1)!$ by Lemma \ref{lem:complexity}.
Depending on the method used for the first step, this second step of finding the dual graph may not be necessary, since some convex hull algorithms produce it as a side product \cite{HDCG:3:convex+hulls,polymake:2017}.

The third step is to find the $\ell$ epistatic weights, each of which is gotten by computing three determinants of size at most $(n+2){\times}(n+2)$.
This adds up to a total cost of $\bigOh(\ell n^3)$.
Sorting the dual edges ascendingly by epistatic weight takes $\bigOh(\ell \log \ell)$.

In the fourth and final step we create the epistatic filtration as a rooted binary tree.
We iterate over the thresholds, which define cluster partitions.
On the way we maintain a forest where each tree represents one cluster in the corresponding partition.
Initially, each tree in the forest is an isolated node, one for each maximal cell of $\subdivision$.
For each dual edge $(s,t)$ with the next epistatic weight $\theta$ we merge clusters of $\subdivision(\theta-\epsilon)$ into clusters of $\subdivision(\theta)$.
Then we remove the trees $T$ and $T'$ containing the leaf nodes corresponding to $s$ and $t$ and add one tree with a new root and $T,T'$ as children. By convention, the left child should always be smaller than the right child with respect to some linear order. In our calculations, we use the lexicographic order on the underlying vertex sets.
The process ends with the dual edge of highest epistatic weight, and we obtain a rooted binary tree that is uniquely determined up to the choice of the order.
The running time of this step is linear in the number of nodes in the resulting tree.
Any binary tree has less than twice as many nodes as leaves; hence the cost adds up to $\bigOh(k)$.

Altogether we arrive at a complexity of
\[
  \bigOh(k^2 n + \ell n^3 + \ell \log \ell + k) \ = \ \bigOh(k^2 n + \ell n^3)
\]
for the steps two through four.
Note that $\ell \log \ell \leq k^2n$ in view of $\ell \leq kn$ by Lemma \ref{lem:complexity}.
The first step, which requires a convex hull computation, is the bottleneck.

We implemented our method in \polymake \cite{DMV:polymake,polymake:2017}, and this was used to produce all computational results presented in this article.




\subsection{An extended example} \label{subsec:extended_example}

To illustrate the concepts described in Section \ref{sec:clusters}, let $P$ be a 3-dimensional cube contained in $[0,1]^5$ with vertex set $V^{(3)}$  given by
\[
  \begin{array}{llll}
    0=0 0 0 0 0\ ; & 1=1 0 0 0 0\ ;& 5=0 0 0 0 1\ ; & 9=1 0 0 0 1\ ; \\
    4=0 0 0 1 0\ ;  &  8=1 0 0 1 0\ ; & 15=0 0 0 1 1\ ;  & 21=1 0 0 1 1 \enspace .
  \end{array}
\]
The vertex labels with their corresponding bit strings (e.g., 4=00010 means that vertex 4 corresponds to bit string 00010) are the genotypes listed in Table~\ref{tab:genotypes} and are viewed 
as points in $\mathbb R^5$.
Consider a height function $\text{\ttd}^{(3)}$ assigning the following values to the eight vertices of~$P$:
\begin{equation}\label{eq:ttd3}
  \begin{array}{llll}
    0 \mapsto 53.25\ ; & 1 \mapsto  46.65\ ; & 5 \mapsto  43.16\ ; & 9 \mapsto  43.48\ ; \\
    4 \mapsto  48.3\ ; & 8 \mapsto 47.79\ ; & 15 \mapsto  43.53\ ; & 21 \mapsto 40.71  \enspace .
  \end{array}
\end{equation}
Specifically, this height function $\text{\ttd}^{(3)}$, is given by restricting the time to death fitness landscape defined over the whole $[0,1]^5$
to the 3-cube with vertices as above. Details about this fitness landscape and others are given in Section \ref{sec:data}.

The induced triangulation 
$
\subdivision:= \subdivision(V^{(3)},\text{\ttd}^{(3)}) 
$
has six maximal cells:
\[
  \begin{array}{lll}
    A =\{\text{0 1 8 9}\} \ ;  & B = \{\text{0 5 9 15}\} \ ; & C = \{\text{0 8 9 15}\} \ ; \\
    D = \{\text{8 9 15 21}\} \ ; & E = \{\text{0 4 8 15}\}\enspace . & \phantom{C = \{\text{0 4 8 15}\}} 
  \end{array}
\]
A basic combinatorial invariant of an $n$-dimensional polyhedral complex is its \emph{$f$-vector} $f=(f_0,f_1,\dots,f_n)$, where $f_k$ is the number of $k$-dimensional cells.
Here we have $f(\subdivision)=(8, 18, 16, 5)$.
For the tight span, which agrees with the dual graph, we get $f(\subdivision^*)=(5, 4)$.

A direct computation shows that the normalized volume of the cell $D$ is two 
whereas all the other cells have normalized volume one.
This makes $ \subdivision$ a non-unimodular triangulation.
Its dual graph $\Gamma(\subdivision)$ is shown in Figure~\ref{fig:ttd_3proj154}.
For each edge in $\Gamma(\subdivision)$ we can compute the epistatic weight using the determinant expression from \eqref{eq:epistatic-weight}.
For instance, the $3$-simplices $C$ and $D$ are adjacent in $ \subdivision$, and we have
\[
  e_{\text{\ttd}^{(3)}}(C,D) \ = \ \bigl|\det E_{\text{\ttd}^{(3)}}(C,D)\bigr| \cdot \frac{\sqrt{3}}{2} \ \approx \ 0.113 \enspace .
\]
This computation reveals that $\text{\ttd}^{(3)}$ almost induces
a linear dependence among the lifted points indexed by vertices of $C$ and $D$, i.e. $\{(v,\text{\ttd}^{(3)}(v)) \vert v\in C\cup D \}.$
The computations in Example \ref{subsect:continuation_ex} 
also provide no evidence against 
the vanishing of $e_{\text{\ttd}^{(3)}}(C,D)$.
Therefore, we can assume that the additive assumption holds, thus 
knowing the fitness of the genotypes belonging to $C$ 
allows us to deduce the fitness of the unique (satellite) genotype of $D$ which is not in $C$.

A similar computation of the epistatic weights of the remaining dual edges yields the epistatic filtration in Table~\ref{fig:ttd_3proj154}.
The rows are sorted with increasing epistatic weight.

\begin{figure}[ht]
  \begingroup
  \newcommand\radius{2.5}
    \includegraphics[width=\textwidth]{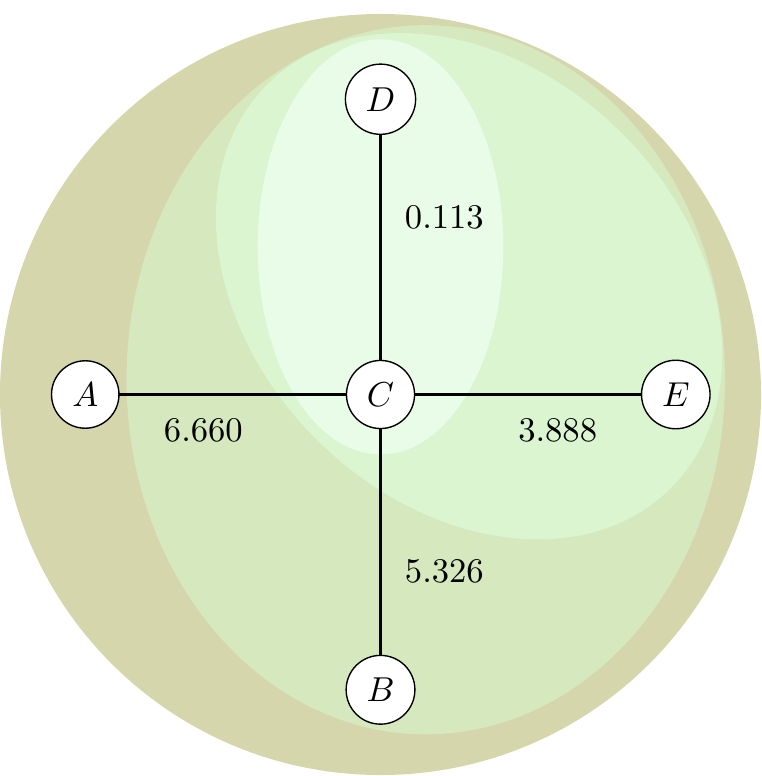} 
  \endgroup
  
  \caption{%
    Dual graph $\Gamma(\subdivision)$ of $ \subdivision$ induced by the restricted height function $\text{\ttd}^{(3)}$.
    Each dual edge is labeled with its epistatic value.
    The shading of the colors indicates the nesting of the cluster partitions.
  }
\label{fig:ttd_3proj154}
\end{figure}

\begin{table}[ht]
  \caption{Epistatic filtration arising from the restricted height function $\text{\ttd}^{(3)}$.
    The color of the circle to the right of each cluster partition agrees with Figure~\ref{fig:ttd_3proj154}.}
  \label{tab:filtrationttd_3proj}
  \begin{tabular*}{.67\linewidth}{@{\extracolsep{\fill}}lclr@{}}
    \toprule
    \multicolumn{1}{c}{$\theta$} & \multicolumn{1}{c}{$(s,t)$} & \multicolumn{1}{c}{$\subdivision(\theta)$} \\
    \midrule
    0  & $-$ & $A|B|C|D|E$ &  \\
    0.113 & $(C,D)$ & $A|B|CD|E$  & \tikz\draw [green0.065] ($(0,0)$) ellipse (0.15 cm and 0.15cm);  \\
    3.888 & $(C,E)$ & $A|B|CDE$  & \tikz\draw [green2.245] ($(0,0)$) ellipse (0.15 cm and 0.15cm);  \\
    5.326 & $(B,C)$ & $A|BCDE$ & \tikz\draw [green3.075] ($(0,0)$) ellipse (0.15 cm and 0.15cm); \\
    6.660 & $(A,C)$ & $ABCDE$ & \tikz\draw [green3.845] ($(0,0)$) ellipse (0.15 cm and 0.15cm); \\	
    \bottomrule
  \end{tabular*}
\end{table}
The initial cluster partition, for $\theta=0$, consists of five connected components, one for each maximal cell of~$\subdivision$:
\[
  \subdivision(0) \ = \ A|B|C|D|E \enspace ,
\]
From Table~\ref{tab:filtrationttd_3proj} we see that, among the four dual edges of $\Gamma(\subdivision)$, the dual edge $(C,D)$ is the one of lowest epistatic weight.
This is the second row of Table~\ref{tab:filtrationttd_3proj} and we have
\[
  \subdivision(0.113) \ = \ A|B|CD|E \enspace .
\]
After three more steps for $\theta=3.888$, 5.326, 6.660 we finally arrive at the trivial cluster partition
\[
  \subdivision(6.660) \ = \ ABCDE \enspace ,
\]
obtained from joining the adjacent simplices $A$ and~$C$.
In this triangulation all dual edges are critical.
That is, the cluster partitions arising from the epistatic weights of all dual edges are pairwise distinct.

\begin{figure}[ht]
  
\includegraphics[width=\textwidth]{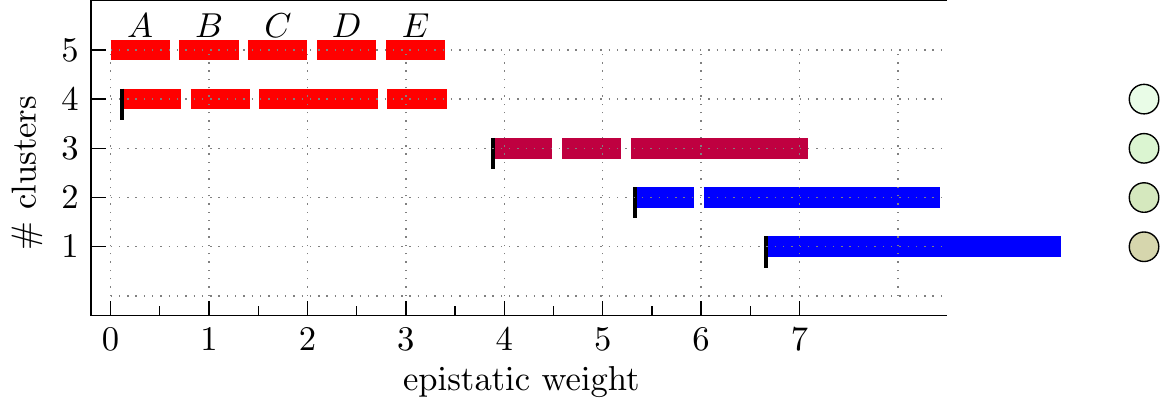}
  
  \caption{Visualizing the epistatic filtration from Table~\ref{tab:filtrationttd_3proj}.
    The black ticks mark the pairs $(\theta,\ell)$, where $\theta$ is a critical threshold, and $\ell$ is the corresponding level.
    The cluster partitions are drawn consistently through all levels.
    The color of the circle to the right of each cluster partition agrees with Figure~\ref{fig:ttd_3proj154}.
    The colors of the bars represent various levels of statistical significance 
    ($p<0.05$: \textcolor{blue}{blue}, $0.05\leq p<0.1$: \textcolor{purple}{purple}, $p\geq 0.1$: \textcolor{red}{red}); cf.\ Section~\ref{sec:statsignificance1}. 
  }
  \label{fig:bar-diagram}
\end{figure}

The epistatic filtration is the sequence of nested cluster partitions which arises from increasing the threshold values.
The whole process can be visualized as follows.
For each critical threshold $\theta$ we mark the point $(\theta,\ell)$ in a planar diagram, where $\ell$ is the \emph{level}, i.e., the number of clusters in the $\theta$-cluster partition.
To the right of the marking $(\theta,\ell)$ we draw the $\ell$ clusters as intervals such that the length of each interval is proportional to the size of the corresponding cluster.
Any two subsequent levels are consistently drawn in the following sense,
suppose that $\theta$ and $\theta'$ are two subsequent critical thresholds; i.e., the open interval $(\theta,\theta')$ is regular.
Let $\ell$ and $\ell'$ be the $\theta$- and $\theta'$-level, respectively.
Then the $\theta$-cluster partition refines the $\theta'$-cluster partition or, conversely, the $\theta'$ cluster partition arises from joining clusters.
That is, $\ell>\ell'$.
To see which clusters get joined one can compare the two sequences of intervals starting from the left (or from the right).
The length of each interval on level $\ell'$ indicates how many clusters of level $\ell$ get joined.
Such a consistent way of drawing an epistatic filtration always exists since the nested cluster partition of all levels form a tree.
By labeling the clusters on the top level, such a diagram encodes the entire epistatic filtration.




\section{Significant cluster partitions} \label{sec:statsignificance1}

The purpose of our epistatic filtrations is to help separate \enquote{biologically interesting} 
epistatic information from noise; a geometric view is expressed in Theorem~\ref{thm:undo}.
Our next goal is to detect cluster partitions which are significant in a statistical sense.
To do so, in Section \ref{subsec:significance_epistaticweight} we develop a 
hypothesis test for edges in the
dual graph $\Gamma(\subdivision).$

\subsection{Error analysis and standard deviation}
Let $P$ and $V$ be as before.
Throughout, for each $v\in V$ let $X_v$ be a positive (absolutely) continuous random variable.
Assume the first two moments of $X_v$ exist and are given by
$\expectation\big(X_v\big)$ and $\sigma_{X_v}=\sqrt{\expectation\big(X_v^2\big)-\big(\expectation(X_v)\big)^2}$. 
We view $X=(X_v)$, for $v\in V$, as a vector of random variables and 
assume that each realization of $X$ is a generic height function on $V$ with probability one.
Let $s=\conv\{v_1,\dots,v_{n+1}\}$ and $t=\conv\{v_2,\dots,v_{n+2}\}$ be two simplices which are spanned by points in $V$ and which share a common codimension-$1$-cell.
These are candidates for two adjacent maximal cells of $\subdivision(V,X)$; such cells are simplices almost surely, as $X$ is generic with probability one.
In addition, let $E_{X}(s,t)$ be the matrix from \eqref{eq:epistatic-matrix} with $h$ replaced by $X$.
We write $E_i=E_{X}(s,t)_{i}$ for the matrix obtained from $E_{X}(s,t)$ by deleting the $i$-th row and the last column.
Note that the coefficients of $E_i$ are ones or coordinates of vertices in $V$.
In particular, those coefficients do not depend on $X$ or any other entries of the last column of the matrix given in \eqref{eq:epistatic-matrix}.

\begin{remark}\label{rem:bipyramid}
  The epistatic weight was defined in the situation where the pair $(s,t)$ forms a dual edge of some subdivision.
  Yet the formula \eqref{eq:epistatic-weight} makes sense even without that assumption, i.e., for an arbitrary bipyramid.
\end{remark}

To simplify the exposition, in the following claim let $N:=\nvol(s\cap t)/(\nvol{s}\cdot \nvol{t})$ and 
let $i,j\in \{1, n+2 \}$.
Let $X\sim\normalDistribution(\mu,\sigma^2)$, then $Y=\vert X \vert\sim\mathcal{F}\normalDistribution(\mu,\sigma^2)$ is
again defined  by $\mu$ and $\sigma^2$ and is called a \emph{folded normal distribution}, 
see \cite{math2010012}.
The density function of $Y$ is given by:
\begin{equation}\label{formula:fnd}
f(y \, \vert \, \mu,\sigma^2)=\frac{1}{\sqrt{2\pi\sigma^2}}\big(e^{-\frac{1}{2\sigma^2}(y-\mu)^2}+e^{-\frac{1}{2\sigma^2}(y+\mu)^2}  \big)\enspace.
\end{equation}

\begin{proposition}\label{prop:se}
  We set \[ \lambda_i \ := \ (-1)^{n+i}\,  N \, \det(E_{i})\enspace .\]
  First, the expectation of the random variable $e_{X}(s,t)$ satisfies
  \[
    \bigg\lvert \ \sum_{i=1}^{n+2} \lambda_i \, \expectation\bigl(X_{v_{i}}\bigr)\bigg \rvert
    \ \leq \ \expectation \big(e_{X}(s,t) \big)
    \ \leq \ 
    \sum_{i=1}^{n+2}\big\lvert \lambda_i\big \rvert \, \expectation \bigl(X_{v_{i}}\bigr) \enspace,
  \]
  and its variance satisfies
  \begin{equation}\label{eq:se-ineq}
    \begin{split}
      \sigma^2_{e_{X}(s,t)}  \ \leq \ & \sum_{i=1}^{n+2} \left(N\, \det(E_{i})  \, \sigma_{X_{v_i}} \right)^2 \\
      & + \ 2\, N^2 \, \sum_{1\leq i < j \leq n+2} \left|\det(E_{i} \cdot E_{j})\right| \,\sigma_{X_{v_i}}\sigma_{X_{v_j}}\enspace .
    \end{split}
  \end{equation}
  Second, if the $n+2$ random variables $X_{v_i}$ are mutually independent then
  \begin{equation}\label{eq:se-indepcase}
    \sigma^2_{e_{X}(s,t)} \ \leq \ \sum_{i=1}^{n+2}\left(N \,\det(E_{i})  \,\sigma_{X_{v_i}} \right)^2 \enspace .
  \end{equation}
  Third, if in addition to independence the relation $X_{v_i}\sim\normalDistribution(\expectation (X_{v_i}),\sigma_{X_{v_i}}^2)$ holds for all~$i$, then  $e_{X}(s,t)$ has folded normal distribution 
  \begin{equation}\label{eq:se-norindepcase}
    \foldedNormalDistribution\left (\sum_{i=1}^{n+2} \lambda_i \, \expectation\big(X_{v_i}\big), \sum_{i=1}^{n+2} \bigl(\lambda_i \, \sigma_{X_{v_i}}\bigr)^2\ \right)
    \enspace.
  \end{equation}
\end{proposition}

\begin{proof} 
For a fixed ordering of the $n+2$ vertices in $s\cup t$, the numbers $\lambda_i$ are real constants. To deduce the first claim, use the Laplace expansion of $E_{X}(s,t)$ along the last column and the linearity of the expected value. The left inequality follows by Jensen's inequality $\vert \expectation (Y) \vert \leq \expectation(\vert Y\vert)$, \cite[Thm.\  B.17 in \S.B.2.2]{schervish1996theory}.
The right inequality follows from the triangle inequality and since $X_{v_i}$ are assumed to be positive.
The approximations of the variance of $e_{X}(s,t) $ in \eqref{eq:se-ineq} follow by the bilinearity of the covariance, Jensen's inequality and the Cauchy-Schwarz inequality \cite[Thm. B.19]{schervish1996theory}.

If all $X_{v_i}$ and $X_{v_j}$ are mutually independent, then $\cov\big(X_{v_i},X_{v_j}\big)=0$ 
in \eqref{eq:se-ineq}, and this  settles~\eqref{eq:se-indepcase} in our claim.

The last claim follows, e.g., from the convolution property of the density functions of $X_{v_i}$,
see \cite[\S. B.1.3]{schervish1996theory} or \cite[Prop.7.17]{dauxois2004toutes}.
Passing to the absolute value implies that the distribution of $e_{X}(s,t)$ is 
the folded normal distribution, defined by the claimed parameters. 
\end{proof}




\subsection{Significance test for the epistatic weight }\label{subsec:significance_epistaticweight}
Assume that the distribution mean $\mu=\expectation\big({e_{X}}(s,t)\big)$ 
is unknown. 
To test if $\mu$ is zero or not, we set up a one-sided test of significance. The null hypothesis is $\mu=0$, and the alternative hypothesis is $\mu> 0.$ 
To define the test statistics, for each $v\in V,$ consider a sample 
${\pmb x}(v)=\big(x_1(v),x_2(v),\dots,x_L(v)\big)$ 
of size $L$ of independent and equally distributed 
realizations of $X_v$.
The number $L$ may vary across $v\in V$.
Let $\bar{x}:V\to\RR$ be defined 
by the sample mean $\bar{x}_v$ at each $v\in V$.
Let $\bar{X}_v$ be the random variable 
evaluating to $\bar{x}.$ 
Since the sample size $L$ is large enough, we assume that $\bar{X}_v\sim\normalDistribution(\bar{x}_v, s_{\bar{x}_v})$, 
for each $v\in V$, and where $s_{\bar{x}_v}=\sqrt{\sum_{i=1}^L(x_i(v)-\bar{x}_v)^2}/\sqrt{L}.$
Let $s$ and $t$ be two adjacent simplices of the triangulation induced by $\bar{x}$.
We define the test statistics for the dual edge $(s,t)$ to be $z=e_{\bar{x}}(s,t)$.
Assuming pairwise independence of $X_v$ for $v\in V$, we deduce that the random variable $Z$ evaluating to $z$ satisfies
$Z\sim\foldedNormalDistribution(0,\sigma_{e_{\bar{X}}(s,t)}^2)$ under the null assumption.

The validity of the null hypothesis is then deduced by computing the $p$-value of the test:
\begin{equation}\label{formula:p}
  P(Z\geq z) \ = \ \int_{z}^{\infty} \frac{\sqrt{2}}{\sigma_{e_{\bar{X}}(s,t)}\sqrt{\pi}} e^{-\frac{1}{2}\big(\frac{y}{\sigma_{e_{\bar{X}}(s,t)}}\big)^2}dy \enspace ,
\end{equation}
where the integrand is the density function given in \eqref{formula:fnd} and $\sigma_{e_{\bar{X}}(s,t)}$
is estimated as in Prop. \ref{prop:se}.
We call the dual edge $(s,t)$ statistically significant if its $p$-value fulfills $p<0.05$.
In this case the null hypothesis can be rejected.
Setting the significance level at $0.05$ is a common choice; cf.\ \cite{MR2723410}.

Naturally higher epistatic weights are more likely to be significant.
However, this does not always have to be the case as also the 
standard deviation of $e_{\bar{X}}(s,t)$ is taken into consideration in this test.

\begin{remark} 
Our assumptions on $Z$ are plausible for the data analyzed in this paper, see Section \ref{sec:data}.
At the same time, these assumptions are permissive in terms of significance.
Moreover, computing epistatic weights can be of interest, if the lifted point configuration is given by mutually independent random variables, as well as correlated random variables.
\end{remark}




\subsection{Clusters from significant epistatic weights}
\label{subsec:significance_clusters}

We now explain how the notion of significant epistatic weights makes 
the cluster filtration process discussed in Section \ref{subsec:clusters} into a
biologically meaningful clustering algorithm.
For this let $\subdivision(V,\bar{X})$ be as above an induced regular triangulation of an $n$-polytope $P$ equipped with a height function $\bar{X}$.
Again we assume that $\bar{X}$ assigns to each vertex of $P$ the sample mean over a number of experimental measurements. 

We saw that varying the parameter $\theta$ partitions $\subdivision$ into clusters ordered according to their epistatic weight.
This can be used to discard noisy signal from relevant data. 
However, that approach does not take into account the dispersion of the experimental measurements.
To account for this, we propose to combine the epistatic filtrations with the significance test discussed above.
More precisely, we suggest to compute all epistatic weights for the dual graph of $\subdivision$, 
contract the edges whose epistatic weight does not reach significance at $p<0.05$ and unify the labels of the affected vertices. 
We call the remaining graph the \emph{significant subgraph} $\significantsubgraph(\subdivision)$ of~$\dualgraph(\subdivision)$.
This induces \emph{significant clusters} and the \emph{significant cluster partition}~$\significantclusterpartition$.
Now the epistatic filtration process from Section~\ref{subsec:clusters} and the algorithm from Section~\ref{subsec:computing} carry over.



\subsection{Continuation of the extended example}\label{subsect:continuation_ex}

We now illustrate the above definitions on the example in \ref{subsec:extended_example}.
Consider again the regular triangulation $\subdivision=\subdivision(V^{(3)},\text{\ttd}^{(3)})$ induced by the height function $\text{\ttd}^{(3)}$ from \eqref{eq:ttd3}.
Figure~\ref{fig:ttd_3proj154} shows the dual graph of $\subdivision$; the five maximal cells are labeled $A,B,C,D,E$.
The value for each vertex is the sample mean over a large number of outcomes of replicated experiments.
Therefore we now write $\bar{X}$ instead of $\text{\ttd}^{(3)}$.
The standard error of the mean for the eight vertices is given by:
 \[
 \begin{array}{llll}
    0 \mapsto 1.450  & 1 \mapsto 1.498 & 5 \mapsto 1.136  & 9 \mapsto  0.988 \\
    4 \mapsto  1.010 & 8 \mapsto 1.098 & 15 \mapsto 1.156  & 21 \mapsto  0.907 \enspace .
 \end{array}
\]
Assuming independence, for each dual edge in $\Gamma(\subdivision)$ we now compute bounds on the associated standard deviation using \eqref{eq:se-indepcase}.
For instance, for $e_{\bar{X}}(C,D)$ consider
\[
  E_{\sigma_{\bar{X}}}(C,D) \ = \
  \begin{pmatrix}
    1 & 0 & 0 &0 & 1.450 \\
    1 & 0 & 1 &1 & 1.098 \\
    1 & 1 & 0 &1 & 0.988 \\
    1 & 1 & 1 &0 & 1.156 \\
    1 & 1 & 1 &1 & 0.907
  \end{pmatrix} \enspace .
\]
This yields
\[
  \begin{split}
    \sigma_{e_{\bar{X}}}&(C,D)\ \leq \ \left(\sum_{i=1}^{5}\left(N \,\det(E_{\sigma_{\bar{X}}}(C,D)_{i}) \,\sigma_{\bar{X}}(v_i) \right)^2\right)^{1/2}\ \approx\ 2.586 
  \end{split}
\]
with $N:= \nvol(C\cap D)/\nvol(C)\cdot \nvol(D)$. Processing the other epistatic weights in a similar fashion yields the values in Table \ref{tab:sdep}.
The rows are sorted by increasing epistatic weight.
The bounds on the $p$-values in the last column are determined by \eqref{formula:p}.

\begin{table}[ht]
  \caption{Significance of epistatic weights. The dual edge in bold is asserted to reach significance for $p<0.05$.}
  \label{tab:sdep}
  \begin{tabular*}{.67\linewidth}{@{\extracolsep{\fill}}lccrl@{}}
    \toprule
    $(s,t)$   & $e_{\bar{X}}(s,t) $ & $\sigma_{e_{\bar{X}}}(s,t)\leq$ & $p\text{-value}\leq$  \\
    \midrule
    $(C,D)$ & 0.113 & 2.586 &  0.965 & \tikz\draw [green0.065] ($(0,0)$) ellipse (0.15 cm and 0.15cm);     \\
    $(C,E)$ & 3.888 & 2.698 &  0.149  & \tikz\draw [green2.245] ($(0,0)$) ellipse (0.15 cm and 0.15cm);   \\
    $(B,C)$ & 5.326 & 2.844 & 0.061  & \tikz\draw [green3.075] ($(0,0)$) ellipse (0.15 cm and 0.15cm);   \\
    $\boldsymbol{(A,C)}$  & 6.660 & 3.309 & 0.044 & \tikz\draw [green3.845] ($(0,0)$) ellipse (0.15 cm and 0.15cm);    \\
    \bottomrule
  \end{tabular*}
\end{table}

In this example, the significant cluster partition $\significantclusterpartition$ 
arising from the restricted height function $\bar{X}$ reads $A|BCDE$.  

\begin{remark}
  Since \eqref{eq:se-indepcase} only provides an upper bound on the $p$-value, it is useful to also investigate dual edges and epistatic weights whose $p$-value bounds are near $0.05$.
  In this way, the dual edge $(B,C)$ with epistatic weight $5.326$ and $p$-value bound $0.061$ comes into focus.
  It would be interesting to check if additional experiments involving the five genotypes $0,5,8,9,15$ in $B\cup C$ lead to a higher level of significance or not.
\end{remark}

\subsection{A synthetic experiment}
\label{subsec:synthetic}


In this section we describe one synthetic experiment on a quantitative analysis of Theorem~\ref{thm:undo} in relationship with the concept of statistical significance from Section~\ref{sec:statsignificance1}. The purpose is to evaluate the extent to which changes in the height function for a single vertex impact the overall epistatic filtration. 

We consider the vertex set $V=\{0,1\}^5$ of the regular $5$-cube and a height function $\eta$ which takes every vertex to height $5$ except for the wild type, which is mapped to $5+\eta_0$ 
for some strictly positive real number $\eta_0$.
The wild type corresponds to the vertex $0$; cf.\ Table~\ref{tab:genotypes}.
The combinatorial type of the regular subdivision $\subdivision(V,\eta)$ does not depend on the precise value $\eta_0>0$.
In fact, there are precisely two maximal cells, $s_0$ and $t$, and $\subdivision(V,\eta)$ is a \emph{vertex split} in the terminology of \cite{HerrmannJoswig:2008}.
The cell $s_0$ is a simplex spanned by the wild type and its five neighbors (i.e., the standard simplex with the origin and the five unit vectors as its vertices).
The other cell, $t$, is not a simplex; instead this is the convex hull of all 31 vertices different from the wild type.
The intersection of $s_0$ and $t$ is the regular $4$-simplex spanned by the five unit vectors. 
So the dual graph $\dualgraph(\subdivision(V,\eta))$ has two nodes connected by the single dual edge $(s_0,t)$.

Our experiments depend on the choice of $\eta_0$ and a second strictly positive real number $\sigma$.
To each vertex $v\in V$ we assign a normally distributed random variable $X_v$ with zero mean and standard deviation $\sigma$.
From $100$ realizations per vertex we compute the resulting sample means and standard errors.
This gives rise to a generic perturbation
\[
  \eta' \ = \ \eta+(\bar{X}_v\,|\,v\in V)
\]
of the height function $\eta$ by adding the sample means.
For the resulting triangulation $\subdivision(V,\eta')$ we compute the epistatic weights and the $p$-values (based on the standard errors computed).
These are all the ingredients required for the significance test via \eqref{eq:se-indepcase}.
We start with a height function $\eta$ at level $5$ since the mean values $\bar{X}_v$ from the perturbation may be negative; note that the perturbed height function $\eta'$ needs to be strictly positive in order to qualify for the analysis via the $p$-values from \eqref{formula:p}.

We only consider perturbed height functions $\eta'$ such that the simplex $s_0$ is a maximal cell of $\subdivision(V,\eta')$, just as in $\subdivision(V,\eta)$.
For $\sigma$ sufficiently small compared to $\eta_0$ this holds almost always.
If $s_0$ is a maximal cell then $s_0$ is adjacent to some unique maximal cell of $\subdivision(V,\eta')$.
We call the corresponding dual edge the \emph{bridge} of $\dualgraph(\subdivision(V,\eta'))$.

Now, for a fixed pair $(\eta_0,\sigma)$ we repeat the above random construction $100$ times, and we count how often the bridge is significant with respect to $p=0.05$ and $p=0.1$.
In all the cases that we saw the simplex $s_0$ was a maximal cell, and the bridge existed.
Further all perturbed height functions $\eta'$ were nonnegative.
Figure~\ref{fig:synth_experiment} shows the result for 
$\eta_0\in\{0.8, 1.0, 1.2\}$ and $0.1 \leq \sigma\leq 2$.

\begin{figure}[ht]

  \centering
  \begingroup
\includegraphics[width=\textwidth]{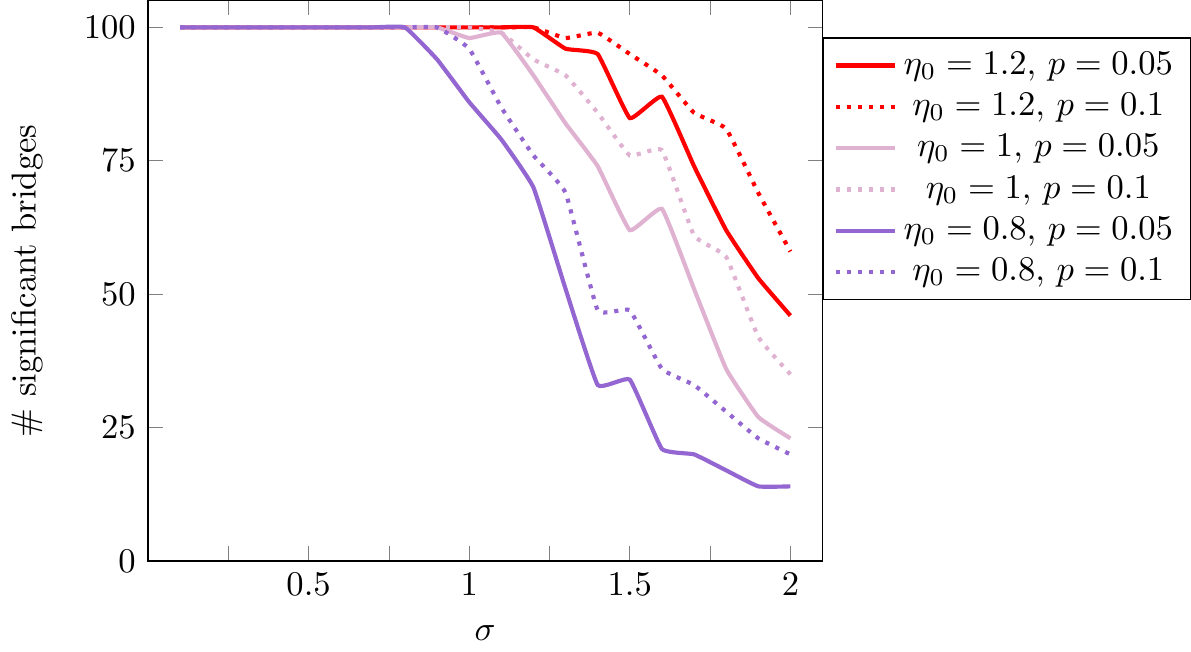}
  \endgroup
  \caption{Percentage of significant bridges depending on $\sigma$, for various choices of $\eta_0$ and $p$.}
  \label{fig:synth_experiment}  
 
\end{figure}

The experimental results displayed in Figure~\ref{fig:synth_experiment} can be summarized as follows:
for any fixed choice of $\eta_0$ and $p$ the percentage of triangulations with the bridge present approximates a threshold function in the parameter~$\sigma$.
If $\sigma$ is sufficiently low then the bridge is always significant; this observation can be seen as a variation of Theorem~\ref{thm:undo} for this particular setup.
The lower the value of $\eta_0$, the steeper the corresponding curve in Figure \ref{fig:synth_experiment}.
For larger values of $\sigma$ random fluctuations kick in, and this makes the curve less smooth.

This experiment provides evidence that our concept of significant cluster partitions is suitable to weed out small statistical fluctuations in the height function.




\section{Epistasis, interaction coordinates and circuit interactions}\label{sec:epistasis}

In this section we compare the cluster partition to previous approaches. To do so, we now recall the biological phenomenon of epistasis as described in the work of Beerenwinkel et al.~\cite{BPS:2007}. 

\subsection{Interaction spaces}\label{sec:int}
Let $P$ be any $n$-dimensional convex polytope with vertex set $V$. Let $\RR^{V}$ be 
the real vector space of all height functions on $V$.
Let $\mathcal{L}_{V}$ be the subspace of $\RR^{V}$ consisting of 
all height functions on $V$ for which the lifted polytope has dimension $n.$
The \emph{interaction space} is the quotient
\[
\mathcal{I}_{V} \ := \ \Bigg(\RR ^{V}/\mathcal{L}_{V} \Bigg)^\ast \enspace.
\] 
Elements of $\mathcal{I}_{V}$ are linear forms
\begin{align*}
\lambda\colon\RR ^{V}&\longrightarrow\RR\\
h&\longmapsto \sum_{v\in V}\alpha_{v}h(v) \enspace ,
\end{align*} 
with vanishing restriction $\lambda\vert_{\mathcal{L}_{V}}$. In the following, we call
elements of  $\mathcal{I}_{V}$ \emph{interactions}. The dimension of the interaction space is
$
\mathrm{dim}(\mathcal{I}_{V})=\mathrm{dim}(\RR^V)- \mathrm{dim}(\mathcal{L}_{V})=\vert V \vert - \mathrm{dim}(P)-1
$.

When $V$ is the vertex set of an $n$-cube, a 
basis for $\mathcal{I}_{V}$ is given by the \emph{interaction coordinates} 
defined up to multiplication by a scalar as:
 \[
   \begin{split}
u_{h,w}:\RR ^{V}&\longrightarrow\RR\\
h&\longmapsto\sum_{v\in V}
(-1)^{ \langle v , w \rangle}h(v) \enspace ,
\label{eq:int}
   \end{split}
 \]
where $v,w\in\{0,1\}^n$ are vertices of $P$ and $w$ is assumed to have at least two coordinates being 1.

When $n=2,$ and a height function is fixed, then $u_{h,11}=\epsilon(00,01,10,11)$ as defined in Example \ref{ex:running_part2}.
A possible generalization of the usual epistasis formula arises by considering
\emph{circuit interactions}, as defined in \cite{BPS:2007}. 
These are linear forms contained in the interaction space with support given by
a minimal affinely dependent set of vertices of $P$. 
Notice that if the height function $h$ induces an affine function on $s\cup t$, then the 
lift of a set of affinely dependent points is also affinely dependent.
Some examples of circuit interactions and possible epistatic interpretations are given
in \cite[Example 3.8, 3.9]{BPS:2007}.

\subsection{Epistatic weight interactions}\label{sect:epint}

In this work, a new distinguished set of interactions inside $\mathcal{I}_{V}$ are given by the linear forms:
 \[
   \begin{split}
     \RR ^{V}&\longrightarrow\RR\\
     h&\longmapsto e_{h}(s,t) \enspace ,
   \end{split}
 \]
where $s$ and $t$ are adjacent $n$-simplices with vertex set in $V$.
When $h$ is generic, $s$ and $t$ are given as maximal adjacent simplices 
in the induced regular triangulation $\subdivision(V,h)$.

When $V=\{0,1\}^2$, there is a unique epistatic weight which agrees with the absolute value of the linear form $u_{h,11}$. 
This is the only case where the epistatic weight agrees with the notion
of interaction coordinates.
Now we will summarize the relation to circuit interactions.

Let $v_1,v_2,\dots,v_{n+2}$ be vertices in $V$ such that $s=\conv\{v_1,\dots,v_{n+1}\}$ and $t=\conv\{v_2,\dots,v_{n+2}\}$ are two $n$-simplices.
Then the $n+2$ vertices in $s\cup t$ are affinely dependent and contain a unique \emph{circuit}, i.e., a minimal affinely dependent set; cf.\ \cite[\S2.4.1]{Triangulations}. Furthermore,  $e_{h}(s,t)\in \mathcal{I}_{V}$ is a circuit interaction in the sense of \cite{BPS:2007}.




\section{Epistasis in \emph{Drosophila melanogaster} fruit fly microbiomes}

In this section, we use the above approach and analyze existing \emph{Drosophila} 
microbiome data. In particular, we demonstrate similarities with and differences from the methods used in \cite{BPS:2007,BMC:2007}
and locate interesting epistatic information where the previous approach is less conclusive.

\subsection{Data}\label{sec:data}

The data we use is published in Tables S1 in \cite[p.30, Supplemental Material(SM)]{Gould232959}. It consists of experimental measurements of \emph{Drosophila} flies inoculated with 
all possible combinations of five bacterial species naturally present in the gut of wild flies. This dataset is remarkably complete as all 32 bacterial combinations are considered.

The data set includes measurements of time to death (days), daily fecundity (progeny/day/female) and development time (days). 
All measurements were repeated many times to give a mean and standard error for each measurement and bacterial combination. More details on the replications of measurements and experimental settings can be found in the Materials and Methods Section of \cite[SM]{Gould232959}.
In this work, we consistently referred to the above fitness landscapes by the labels  {\ttd}, {\fec} and {\dev}. 
When restricting one of these fitness landscapes
to smaller sub-genotopes, we add a superscription, as in Example \ref{subsec:extended_example}.

\subsection{Epistatic weight approach}

With the approach developed in this paper one asks for general epistatic information in genotype--phenotype mappings. 
Contrary to previous studies, epistasis here is understood as a general deviation from additivity rather than a specific manifestation of it, 
quantified for instance by marginal, conditional, $2$-, $3$-, $4$- or $5$-way interactions, see \cite{BPS:2007} for the terminology.

In the specific \emph{Drosophila} data set, the methods developed in this work allowed us to distinguish certain bacterial combinations with vanishing epistatic weights from statistically significant ones. Bacterial combinations with low epistatic weights are not expected to have synergistic or antagonistic interactions between the species, while such effects are expected for bacterial combinations with statistically significant epistatic weight. Filtrations then provide clusters of bacterial combinations, based on adjacent relationships between simplices in a triangulation. Clusters thus determine interesting regions inside the genotope, which we propose should be the targets for further analysis. An example of such an analysis is given in Section \ref{sect:entire_cube}.

Our method intentionally involves a relatively small number of tests, limited by the adjacency relations among maximal simplices in the triangulated genotope and are specific to each phenotype mapping. In Section \ref{sec:statsignificance1} we provided
statistical tools for the data analysis.
These tools also facilitate comparisons with previous approaches. 
In \cite{Gould232959} a number of specific epistatic formulas were
classified as significant. 
These epistatic formulas include standard tests, contextual tests,  interaction coordinates and circuits, as described in \cite[\S 5, Math Supplement, SM, pp.62]{Gould232959}. Significance was 
tested as described in \cite[\S 6, Math Supplement, pp.69 SM]{Gould232959}. The results of
these tests and their significance was reported in Figure 4 A-E, \cite[Main text]{Gould232959}.

Comparing our work with the results of these tests reveals important differences between the epistatic formulas of \cite{Gould232959} and the epistatic filtrations examined here. Examples are given below.

\subsection{The case of the entire $[0,1]^5$}\label{sect:entire_cube}

Filtrations for the fitness landscapes of $[0,1]^5$ defined for the time to death (\ttd), daily fecundity (\fec) and development time (\dev) provide insights on the different cluster patterns
arising but revealed a unique significant epistatic weights. 
The significant dual edge in {\ttd} is given by $e_\ttd(s,t)\approx 5.435$
with a $p$-value of approximatively 0.0376.
This epistatic weight arises over the bipyramid with vertex set
$s=\{0, 9, 12, 14, 15, 28\}$ and
$t=\{0, 9, 12, 14, 27, 28\}$.

This fact has two biological implications. First, it shows
that we have evidence against one 5-dimensional affine
relation between the {\ttd} measurements for $15$ and $27$, 
given the {\ttd} measurements for $s\cup t\backslash\{15,27\}$. 

Second, we observe that
$s\cup t=\{ 0,9,12,14,15, 27, 28\}$
is not a minimal affinely dependent set. Yet, the subset \{14,15,27,28\} is a minimal affinely dependent set, which gives rise to the linear form
\begin{equation} \label{eq:circuit_pq}
\omega_h =\  \big(h(00101) - h(00011)\big) - \big(h(11101) - h(11011)\big) \enspace.
\end{equation} 
This linear form cannot be written as a 3-cube circuit interaction, and allows
for a new biological interpretation.
More specifically, the linear form of equation \eqref{eq:circuit_pq}, 
compares the effect of two pairs of \emph{Acetobacter} 
bacteria (\emph{pasteurianus} with \emph{orientalis}, 
resp.\ \emph{tropicalis} with \emph{orientalis}) 
in the joint presence, resp.\ absence, of both
\emph{Lactobacillus} bacteria. 
Evaluating equation \eqref{eq:circuit_pq} at \text{ttd} gives 
$\vert {\omega_\ttd}\vert=\frac{1}{2}\vert\det E_{\text{ttd}}(s,t)\vert \approx 4.86$ with a
$p$-value of approximatively $0.038$. Thus, there is evidence to believe that
the above form of marginal epistasis is non-additive. This fact refines our first conclusion.
By the discussion in Section \ref{sect:epint}, there is no
other circuit interaction on $s\cup t$.

Significant outcomes for other linear forms on $[0,1]^5$
studied in the context of epistasis are also possible.
For instance, results of recent work imply that for {\ttd}, 124 tests out of 936 resulted to be significantly different than zero, ($p< 0.05$), see \cite[\S 5, Math Supplement, SM]{Gould232959}. These 936 tests include: epistatic weights on all 2-faces in $[0,1]^5$, all 20 circuit interactions, $a$-$\ell$ in \cite{BPS:2007}, for all 3-faces in $[0,1]^5$, and all interaction coordinates for the $k$-faces of $[0,1]^5$, for $k\in\{3,4,5\}$, defined as in Section \ref{sec:int}. 

\begin{remark} In \cite[\S 6, Math Supplement, pp.69 SM]{Gould232959} the discovery rate was corrected by Benjamini--Hochberg
multiple testing correction method \cite{10.2307/2674075}. Correction methods of this type 
aim at decreasing the number of significant outcomes by varying the $p$-values of the tests.
They are typically used when a large number of tests are made, 
making false discoveries more likely.
Due to the fact that there are very few significant epistatic 
weights, here we refrained from applying similar correction methods.
\end{remark}

\subsection{The case of parallel facets inside $[0,1]^5$}
\label{subsec:parallel-facets}

Consider the interaction coordinate:
\begin{equation}\label{eq:u4sign}
  \begin{split}
    u_{h,0\ast101}&\ =\ h({0\ast000})+h({0\ast010})+h({1\ast000})+h({0\ast101})\\ 
    &\ +\ h({1\ast010})+h({0\ast111})+h({1\ast101})+h({1\ast111})\\
    &\ - \ h({0\ast001})-h({0\ast100})-h({0\ast011})-h({1\ast001})\\
    &\ - \ h({1\ast100})-h({0\ast110})-h({1\ast011})-h({1\ast110}) \enspace,
  \end{split}
\end{equation}
where $\ast\in\{0,1\}$.  
The two values of $\ast$ determine so called \emph{four-way interactions}, \cite{BPS:2007},
on parallel facets of $[0,1]^5$.
If $\ast=0,$ the above interaction is considered in the 
absence of the \emph{Lactobacillus brevis} bacteria, 
otherwise the bacteria are present.
The computations of \cite[\S 5, Math Supplement, SM, pp.62]{Gould232959} determine that 
$u_{h,0\ast101}$ is simultaneously significant 
for $h$ given by {\fec}, {\dev} and {\ttd}, 
only in the absence of the \emph{Lactobacillus brevis} bacteria.

\input{diagrams_distinguished_fourfaces}

To further inspect the effect of inoculating 
\emph{Lactobacillus brevis} bacteria in \emph{Drosophilas},
we compute the fitness landscape
defined by the vertices in the summands of $u_{h,0\ast101}$ for both
values of $\ast$ and for $h$ given by {\fec}, {\dev} and {\ttd}. Results are shown in Figure \ref{fig:distinguished}.
For {\ttd} and {\fec}, different cluster patterns appear on the parallel facets. 
The presence of significant epistatic weights in the absence of the \emph{Lactobacillus brevis} 
bacteria confirms that these bacteria significantly affect epistatic interactions.  
Further dissecting the significant epistatic weights, as above, as well as the
filtration steps, restricts the set of possible bacterial 
combinations responsible for this effect.
For {\dev}, all epistatic weights are near zero.
As a consequence, this phenotype produced no 
significant dual edges (all bars are red in Figure 5). We conclude
that the epistatic filtration effect of \emph{Lactobacillus brevis} 
on the \emph{Drosophila} microbiome is most pronounced 
for the phenotypes of {\fec} and {\ttd}.

This result is in contrast to the local significance tests 
computed in \cite[\S 5, Math Supplement, SM, pp.62]{Gould232959}.
Out of those tests, 15 interactions can be seen to be simultaneously
significant for {\fec} and {\dev}. Thus, the two approaches to discover epistatic 
interactions reveal different biological insights, and it remains for future empirical investigations to determine which approach best captures the underlying biological phenomena. 

\subsection{The case of three-cubes inside $[0,1]^5$}\label{sub:cube}

\begin{table}[bh]
  \caption{Circuit interactions and interaction 
coordinates for $\text{\ttd}^{(3)}$ on the 3-cube of Example
   \ref{subsec:extended_example}
   which reached significance for $p<0.05$.}
  \label{tab:scirc}
  \begin{tabular*}{.87\linewidth}{@{\extracolsep{\fill}}lrccl@{}}
    \toprule
    Interaction   &$\vert Z \vert$& $\sigma_{Z}^2$ & $p\text{-value}$ folded & Vertices\\
    \midrule
$a$     &	$6.09$	&	2.57		&0.02 & 0 1 4 8\\
$c$      &	$6.92$ &	2.58		&0.01 & 0 1 5 9\\
$e$	&	$5.32$	&	2.40		&0.03 & 0 4 5 15\\	
$m$	&    $9.10$	&	3.72		&0.01 & 0 1 4 5 21\\
$\varsigma$	&	$7.69$	&	3.83		&0.04 & 0 1 8 9 15\\
$u_{111}$	& 	$9.23$	&	3.32		&0.01 & 0 1 5 9 4 8 15 21\\
    \bottomrule
  \end{tabular*}
\end{table}

To make the comparison between the two methods more explicit, 
consider again Example \ref{subsec:extended_example}.
As before, let $\text{\ttd}^{(3)}$ denote the time to death fitness landscape 
restricted to the genotypes defining the 3-dimensional cube
in $[0,1]^5$ with vertex set  $V^{(3)}=\{0,1,4,5,8,9,15,21\}$.
On the one hand, following \cite{BPS:2007} we know that 
there are 20 circuit interactions of interest and four interaction 
coordinates for such a 3-cube.
Out of these 24 tests, 
six reached statistical significance for $p<0.05$ and assuming that the test statistic satisfies 
$\vert Z \vert \sim\foldedNormalDistribution(\mu,\sigma^2_Z)$.
These significant tests are reported in Table \ref{sub:cube}.
Using the terminology of \cite{BPS:2007},
these significant tests capture three forms of conditional epistasis (for $a,c,e$), 
an interaction coordinate ($u_{111}$, the three-way interaction) 
and the circuit interactions ($m,\varsigma$, two bipyramids in the 3-cube).

On the other hand, out of the four epistatic weights for 
$\subdivision(V^{(3)},\ttd^{(3)})$, computed
in Example \ref{subsec:extended_example}, only $e_{\ttd^{(3)}}(A,C)$ 
reached statistical significance; cf.\ Section \ref{subsect:continuation_ex}.
As before, we have
\[
\varsigma_{\text{\ttd}^{(3)}}\ = \ \vert \det E_{\ttd^{(3)}}(A,C)\vert\enspace.
\]

\subsection{Parallel epistatic weights}

In Section~\ref{subsec:parallel-facets} we studied epistatic effects 
arising from the presence or the absence of one type of bacteria. 
Here we discuss a different way to address the same.

The presence or the absence of the $k$th type of bacteria defines one pair of parallel facets, which induce a partition on the full vertex set.
We denote these facets as $F={\ast}\cdots{\ast}{0}{\ast}\cdots{\ast}$ and $F'={\ast}\cdots{\ast}{1}{\ast}\cdots{\ast}$, respectively; the $0$ and $1$ are in the $k$th position.
The map, $\phi$, which sends a vertex $v=(v_1v_2\dots v_n)\in\{0,1\}^n$ to $v'=(v_1'v_2'\dots v_n')$ where
\[
  v_i' \ = \ \begin{cases} 1-v_k & \text{if $i=k$}\\ v_i & \text{otherwise} \end{cases}
\]
is a reflection which exchanges $F$ and $F'$.
Restricting any generic height function on the entire cube $[0,1]^n$ to the two facets induces two triangulations, of $F$ and $F$', respectively.
Now we can compare the epistatic weights of the bipyramids arising from the triangulation of $F$ with the epistatic weights of their images in $F'$ under the reflection map $\phi$.
We call these \emph{parallel epistatic weights}.
Note that the bipyramids in $F'$, in general, are not bipyramids of the triangulation induced on $F'$.
Still we can compute their epistatic weights, assess their statistical significance and compare; cf.\ Remark~\ref{rem:bipyramid}.
Note that this also applies to any face of the $n$-cube, as faces of cubes are cubes.

\begin{example}
  In Example \ref{subsec:extended_example}, we investigated the epistatic filtration arising from restricting the height function {\ttd} to the $3$-face ${\ast}{0}{0}{\ast}{\ast}$.
  Now we can view ${\ast}{0}{0}{\ast}{\ast}$ as a facet of the $4$-cube ${\ast}{\ast}{0}{\ast}{\ast}$, and ${\ast}{1}{0}{\ast}{\ast}$ is the parallel $3$-face.
  Its vertices are
  \[
    \begin{array}{llll}
      2=0 1 0 0 0\ ; & 6=1 1 0 0 0\ ;& 12=0 1 0 0 1\ ; & 18=1 1 0 0 1\ ; \\
      11=0 1 0 1 0\ ;  &  17=1 1 0 1 0\ ; & 24=0 1 0 1 1\ ;  & 28=1 1 0 1 1 \enspace .
    \end{array}
  \]
  Restricting the height function {\ttd} yields
  \begin{equation}\label{eq:ttd3_p}
    \begin{array}{llll}
      2 \mapsto 52.175\ ; & 6 \mapsto  50.81\ ; & 12 \mapsto  46.79\ ; & 18 \mapsto  43.1\ ; \\
      11 \mapsto  45.46\ ; & 17 \mapsto 43.37\ ; & 24 \mapsto 44.97\ ; & 28 \mapsto 46.15  \enspace .
    \end{array}
  \end{equation}
Thus with $k=2$, the bipyramid $(C,D)$ in the triangulation $\subdivision(V^{(3)},\text{\ttd}^{(3)})$ of ${\ast}{0}{0}{\ast}{\ast}$ gets reflected onto a bipyramid with vertex set
$\{ 2,17,18, 24, 28\}$. The epistatic weight of this bipyramid restriction to the $3$-cube, i.e., after deleting the second and third coordinate, is the parallel epistatic weight of $e_{\text{\ttd}^{(3)}}(C,D)$ for the 3-face ${\ast}{1}{0}{\ast}{\ast}$. It is given by
\begin{equation}\label{eq:det}
  \left|\det \begin{pmatrix}
    1 & 0 & 0 &0 & 52.175 \\
    1 & 0 & 1 &1 & 43.37 \\
    1 & 1 & 0 &1 & 43.1 \\
    1 & 1 & 1 &0 & 44.97 \\
    1 & 1 & 1 &1 & 46.15
  \end{pmatrix} \right| \cdot \frac{\sqrt{3}}{2} \ \approx \ 11.289 \enspace .
\end{equation}

  The epistatic filtration of \ttd$^{(3)}$ and the parallel epistatic weights in ${\ast}{1}{0}{\ast}{\ast}$ are visualized in Figure \ref{fig:parallel_transport_ttd154}, where the bipyramid $(C,D)$ is processed in level $4$.  
Note that the image on the left of that same figure is a horizontally squeezed version of Figure~\ref{fig:bar-diagram}.


\begin{figure}[tp]
  \newcommand\maxepi{15}
  \newcommand\maxclusters{5}
  \newcommand\clusterlength{0.6}
  \newcommand\clusterspace{0.1}
  \newcommand\lw{6pt}
  \centering
  \begin{minipage}[t]{.45\linewidth}
    \centering
	\includegraphics[width=\textwidth]{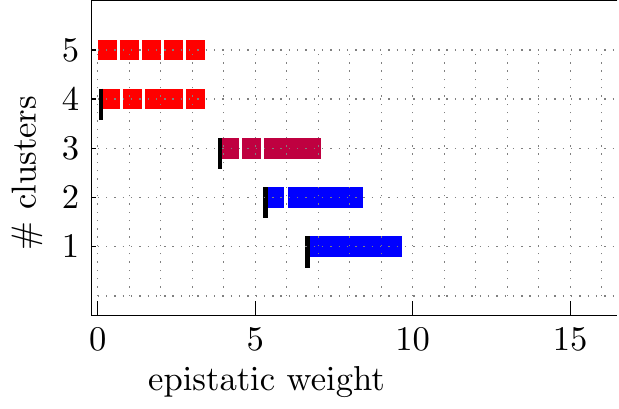}	
  \end{minipage}%
  \hfill%
  \begin{minipage}[t]{.45\linewidth}
    \centering
    \includegraphics[width=\textwidth]{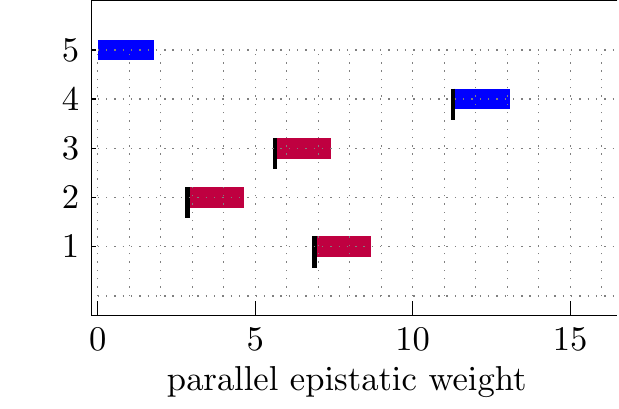}
  \end{minipage}

  \caption{Left: Epistatic filtration of the triangulation $\subdivision(V^{(3)},\text{\ttd}^{(3)})$ of the $3$-cube ${\ast}{0}{0}{\ast}{\ast}$.
    Right: Parallel epistatic weights in ${\ast}{1}{0}{\ast}{\ast}$. 
    The colors of the bars represent various levels of statistical significance 
    ($p<0.05$: \textcolor{blue}{blue}, $0.05\leq p<0.1$: \textcolor{purple}{purple}, $p\geq 0.1$: \textcolor{red}{red}); cf.\ Section~\ref{sec:statsignificance1}. 
 }
  \label{fig:parallel_transport_ttd154}
\end{figure}
Figure~\ref{fig:paralleltransport-facet2} shows (the epistatic filtration of) the triangulation of ${\ast}{0}{\ast}{\ast}{\ast}$ induced by {\ttd} (upper left), the parallel epistatic weights in ${\ast}{1}{\ast}{\ast}{\ast}$ (upper right), the triangulation of ${\ast}{1}{\ast}{\ast}{\ast}$ (lower right) and the parallel epistatic weights in ${\ast}{0}{\ast}{\ast}{\ast}$.

Computing circuit interactions and interaction coordinates, as describe in 
Section \ref{sub:cube}, on ${\ast}{1}{0}{\ast}{\ast}$ yields the 
significant results reported in Table \ref{tab:scirc2}. 
As above,
there is a unique circuit interaction with support given by the points $\{2,17,18,24,28\}$.
Its projection to the 3-cube is given by 
\[n_{h} \ = \ h(011) + h(101) + h(110) - h(000) - 2h(111) \enspace .
\]
Evaluating $n_{h}$ at \ttd$^{(3)}$ 
agrees with the determinant of Equation \eqref{eq:det}, up to sign.
\end{example}

\begin{table}[bh]
  \caption{Parallel circuit interactions and interaction 
coordinates for $\text{\ttd}^{(3)}$ on the 3-face of
${\ast}{1}{0}{\ast}{\ast}$ which reached significance for $p<0.05$.}
  \label{tab:scirc2}
  \begin{tabular*}{.87\linewidth}{@{\extracolsep{\fill}}lrccl@{}}
    \toprule
    Interaction   &$\vert Z \vert$& $\sigma_{Z}^2$ & $p\text{-value}$ folded & Vertices\\
    \midrule
$b$	&4.87 	&	2.23		&0.029& 12 18 24 28 \\
$e$	&4.90	&	2.50		&0.050& 2 11 12 24\\
$f$	&10.49	&	2.68		&0.000& 6 17 18 28\\
$i	$	&9.77	&	2.66		&0.000& 2 11 17 28 \\
$j	$	&5.62	&	2.52		&0.026& 6 12 17 24\\
$k$	&8.17	&	2.60		&0.002& 2 12 17 28 \\
$l$	&7.22	&	2.58	 	&0.005& 6 11 18 24 \\
$n$	&13.04	&	3.41		&0.000& 2 17 18 24 28\\
$o$	&8.89	&	3.50		&0.011& 6 11 12 17 28\\
$q$	&12.09	&	3.66		&0.001& 6 11 12 18 28\\
$u_{011}$	&15.39	&	3.66		&0.000& 2 6 11 12 17 18 24 28\\
    \bottomrule
  \end{tabular*}
\end{table}

\subsection{The case of parallel squares inside $[0,1]^5$}

Fixing two bacterial species, $\alpha$ and $\beta$, defines a $2$-dimensional face of $[0,1]^5$, i.e., a square.
This is the set of four points in $\{0,1\}^5$ where the $\alpha$- and $\beta$-coordinates vary, and all others are set to zero.
Now a third bacterial species, $\gamma$, defines a \emph{parallel square} in $[0,1]^5$, where the $\alpha$- and $\beta$-coordinates vary, and the $\gamma$-coordinate is set to one.
Altogether, $\alpha$, $\beta$ and $\gamma$ define a $3$-dimensional face of $[0,1]^5$.
We want to investigate the impact of the presence of $\gamma$ on the epistatic interaction between $\alpha$ and $\beta$.
This amounts to comparing the epistatic weights of the two subdivisions induced on the parallel pair of squares.

Empirically, we know that the \emph{Lactobacillus} bacteria compete with one another and \emph{Acetobacters} also compete with one another, while \emph{Acetobacters} form mutualist relationships with the \emph{Lactobacilli}. We therefore focus on the combinations $\alpha\in\{1,2\}$, $\beta\in\{3,4,5\}$ and $\gamma\in\{1,2,3,4,5\}-\{\alpha,\beta\}$.
This means that, for each height function, we obtain $2 \cdot 3 \cdot (5-2) = 18$ faces which are spanned by $\alpha$, $\beta$ and $\gamma$.
Each face is a row in each of the three Tables \ref{tab:squares-ttd}, \ref{tab:squares-fec} and \ref{tab:squares-dev} in the appendix; these tables show the result of our analysis for the three height functions \ttd,  \dev~ and \fec, respectively.
The results can also be compared with \cite[Figure S.13, p20, SM]{Gould232959}.

Each table shows the normalized volumes of the $3$-dimensional simplices obtained from the four vertices of each square lifted by the respective height function; here $\omega_0$ is the normalized volume of the simplex arising from the square with $\gamma$-coordinate zero, and $\omega_1$ corresponds to $\gamma$-coordinate equal to one.
The epistatic weights of the single dual edges of the two induced subdivisions on the two parallel squares are the absolute values $|\omega_0|$ and $|\omega_1|$.
Yet, in order to track the effect of adding $\gamma$ to $\alpha$ and $\beta$, here we take the orientation into account.
The sign entry is positive if qualitatively the epistasis is the same, and it is negative if the effect gets reversed.
The final column lists the relative increase or decrease (multiplicatively); i.e., a value of $1$ would mean that the effect stays the same while larger values mean that the effect gets stronger.

Here are a few observations that we find particularly noteworthy.
\begin{enumerate}
\item There are some cases where adding $\gamma$ seems to almost annihilate the epistasis, meaning that bystander species can disrupt an interacting pair.
\item A distinguished combination appears to be $(\alpha,\beta)=(2,4)$:
  for which the time to death and the development time  results match, suggesting the interaction affects these two traits in the same manner.
\item Another interesting case is $(\alpha,\beta)=(1,5)$:
   adding $\gamma$ weakens the epistasis for both time to death and development time. 
\end{enumerate}
So far, we have not determined biological explanations for these facts, however, the analysis defines the direction of our ongoing experimental investigations.
Our results underscore the importance of context in determining microbiome interactions, which is to say that interacting groups of species care a great deal about their neighbors.

\section{Discussion and outlook}

In this paper, we develop a new approach to studying properties of fitness landscapes via regular 
subdivisions of convex polytopes, building on and extending previous work of Beerenwinkel et al.\ (2007). Our approach offers a concise combinatorial way of processing and clustering epistatic information in higher dimensions and is based on statistical principles. 
In Theorem \ref{thm:undo} we present first provable robustness considerations.
The main new tools we propose are cluster partitions and cluster filtrations in weighted graphs associated to fitness landscapes. 
To show that our methods are capable of quantifying epistasis with a higher resolution than previously done, we investigated an existing \emph{Drosophila} microbiome data set. Among other results, we provide new forms of epistasis together with their biological interpretations.

To biologically validate our first promising findings, more biological properties of coexisting bacterial species, completing for instance Table \ref{tab:genotypes}, have to be determined. Moreover, further research has to be undertaken to connect the approaches discussed in this paper to a variety of other methodologies used within this theory, summarized for instance in \cite{Krug}.
Finally, although our methods are fully scalable, it is a matter of time until potential computational bottlenecks can fully be addressed.

\section*{Acknowledgement} 
The data was collected by Alison Gould and Vivian Zhang in Will Ludington's Lab at UC Berkeley and is now published in \cite{Gould232959}.
We are indebted to Christian Haase and Günter Rote for a fruitful discussion concerning epistatic weights, leading to the definition \eqref{eq:epistatic-weight}.




\appendix



\begin{table}[hp]
 \caption{Parallel squares in $[0,1]^5$ for \ttd. }
 \label{tab:squares-ttd}
 
 \begin{tabular*}{\linewidth}{@{\extracolsep{\fill}}lllll@{}}
   \toprule
   \multicolumn{1}{c}{$\alpha=1,~\beta=3$} & \multicolumn{1}{r}{$\omega_{0}$} & \multicolumn{1}{r}{$\omega_{1}$}& 
   \multicolumn{1}{c}{sign}&
   \multicolumn{1}{c}{$|$quot$|$}\\  
   \midrule
   \multicolumn{1}{c}{$\gamma$=2} & \multicolumn{1}{c}{$3.560$} & \multicolumn{1}{r}{$0.485$}& \multicolumn{1}{c}{+} & \multicolumn{1}{c}{0.136}\\
   
   \multicolumn{1}{c}{$\gamma$=4} & \multicolumn{1}{c}{$3.560$} & \multicolumn{1}{r}{$-1.590$}& \multicolumn{1}{c}{$-$} & \multicolumn{1}{c}{0.447}\\

   \multicolumn{1}{c}{$\gamma$=5} & \multicolumn{1}{c}{$3.560$} & \multicolumn{1}{r}{$-4.500$}& \multicolumn{1}{c}{$-$} & \multicolumn{1}{c}{1.264}\\[1.5em]

   \midrule
   \multicolumn{1}{c}{$\alpha=1,~\beta=4$} & \multicolumn{1}{r}{$\omega_{0}$} & \multicolumn{1}{r}{$\omega_{1}$}& 
   \multicolumn{1}{c}{sign}&
   \multicolumn{1}{c}{$|$quot$|$}\\  
   \midrule
   \multicolumn{1}{c}{$\gamma$=2} & \multicolumn{1}{c}{$6.090$} & \multicolumn{1}{r}{$-0.725$}& \multicolumn{1}{c}{$-$} & \multicolumn{1}{c}{0.119}\\

   \multicolumn{1}{c}{$\gamma$=3} & \multicolumn{1}{c}{$6.090$} & \multicolumn{1}{r}{$0.940$}& \multicolumn{1}{c}{+} & \multicolumn{1}{c}{0.154}\\

   \multicolumn{1}{c}{$\gamma$=5} & \multicolumn{1}{c}{$6.090$} & \multicolumn{1}{r}{$-3.140$}& \multicolumn{1}{c}{$-$} & \multicolumn{1}{c}{0.516}\\[1.5em]

   \midrule
   \multicolumn{1}{c}{$\alpha=1,~\beta=5$} & \multicolumn{1}{r}{$\omega_{0}$} & \multicolumn{1}{r}{$\omega_{1}$}& 
   \multicolumn{1}{c}{sign}&
   \multicolumn{1}{c}{$|$quot$|$}\\  
   \midrule
   \multicolumn{1}{c}{$\gamma$=2} & \multicolumn{1}{c}{$6.920$} & \multicolumn{1}{r}{$-2.325$}& \multicolumn{1}{c}{$-$} & \multicolumn{1}{c}{0.336}\\

   \multicolumn{1}{c}{$\gamma$=3} & \multicolumn{1}{c}{$6.920$} & \multicolumn{1}{r}{$-1.140$}& \multicolumn{1}{c}{$-$} & \multicolumn{1}{c}{0.165}\\

   \multicolumn{1}{c}{$\gamma$=4} & \multicolumn{1}{c}{$6.920$} & \multicolumn{1}{r}{$-2.310$}& \multicolumn{1}{c}{$-$} & \multicolumn{1}{c}{0.334}\\[1.5em]

   \midrule
   \multicolumn{1}{c}{$\alpha=2,~\beta=3$} & \multicolumn{1}{r}{$\omega_{0}$} & \multicolumn{1}{r}{$\omega_{1}$}& 
   \multicolumn{1}{c}{sign}&
   \multicolumn{1}{c}{$|$quot$|$}\\  
   \midrule
   \multicolumn{1}{c}{$\gamma$=1} & \multicolumn{1}{c}{$0.715$} & \multicolumn{1}{r}{$3.790$}& \multicolumn{1}{c}{$+$} & \multicolumn{1}{c}{5.301}\\

   \multicolumn{1}{c}{$\gamma$=4} & \multicolumn{1}{c}{$0.715$} & \multicolumn{1}{r}{$1.620$}& \multicolumn{1}{c}{$+$} & \multicolumn{1}{c}{2.266}\\
   
   \multicolumn{1}{c}{$\gamma$=5} & \multicolumn{1}{c}{$0.715$} & \multicolumn{1}{r}{$8.090$}& \multicolumn{1}{c}{$+$} & \multicolumn{1}{c}{11.315}\\[1.5em]

   \midrule
   \multicolumn{1}{c}{$\alpha=2,~\beta=4$} & \multicolumn{1}{r}{$\omega_{0}$} & \multicolumn{1}{r}{$\omega_{1}$}& 
   \multicolumn{1}{c}{sign}&
   \multicolumn{1}{c}{$|$quot$|$}\\  
   \midrule
   \multicolumn{1}{c}{$\gamma$=1} & \multicolumn{1}{c}{$1.765$} & \multicolumn{1}{r}{$8.580$}& \multicolumn{1}{c}{$+$} & \multicolumn{1}{c}{4.861}\\

   \multicolumn{1}{c}{$\gamma$=3} & \multicolumn{1}{c}{$1.765$} & \multicolumn{1}{r}{$2.670$}& \multicolumn{1}{c}{$+$} & \multicolumn{1}{c}{1.513}\\

   \multicolumn{1}{c}{$\gamma$=5} & \multicolumn{1}{c}{$1.765$} & \multicolumn{1}{r}{$2.190$}& \multicolumn{1}{c}{$+$} & \multicolumn{1}{c}{1.241}\\[1.5em]

   \midrule
   \multicolumn{1}{c}{$\alpha=2,~\beta=5$} & \multicolumn{1}{r}{$\omega_{0}$} & \multicolumn{1}{r}{$\omega_{1}$}& 
   \multicolumn{1}{c}{sign}&
   \multicolumn{1}{c}{$|$quot$|$}\\  
   \midrule
   \multicolumn{1}{c}{$\gamma$=1} & \multicolumn{1}{c}{$4.705$} & \multicolumn{1}{r}{$-4.540$}& \multicolumn{1}{c}{$-$} & \multicolumn{1}{c}{0.965}\\

   \multicolumn{1}{c}{$\gamma$=3} & \multicolumn{1}{c}{$4.705$} & \multicolumn{1}{r}{$-2.670$}& \multicolumn{1}{c}{$-$} & \multicolumn{1}{c}{0.567}\\

   \multicolumn{1}{c}{$\gamma$=4} & \multicolumn{1}{c}{$4.705$} & \multicolumn{1}{r}{$4.280$}& \multicolumn{1}{c}{$+$} & \multicolumn{1}{c}{0.910}\\
   
   \bottomrule    
    
  \end{tabular*}
\end{table}

\begin{table}[hp]
 \caption{Parallel squares in $[0,1]^5$ for \dev.}
  \label{tab:squares-dev}

  \begin{tabular*}{\linewidth}{@{\extracolsep{\fill}}lllll@{}}
    \toprule
    \multicolumn{1}{c}{$\alpha=1,~\beta=3$} & \multicolumn{1}{r}{$\omega_{0}$} & \multicolumn{1}{r}{$\omega_{1}$}& 
    \multicolumn{1}{c}{sign}&
    \multicolumn{1}{c}{$|$quot$|$}\\  
    \midrule
    
\multicolumn{1}{c}{$\gamma$=2} & \multicolumn{1}{c}{0.712} & \multicolumn{1}{r}{0.083}& \multicolumn{1}{c}{+} & \multicolumn{1}{c}{0.117}\\

\multicolumn{1}{c}{$\gamma$=4} & \multicolumn{1}{c}{0.712} & \multicolumn{1}{r}{$-0.208$}& \multicolumn{1}{c}{$-$} & \multicolumn{1}{c}{0.293}\\

\multicolumn{1}{c}{$\gamma$=5} & \multicolumn{1}{c}{0.712} & \multicolumn{1}{r}{$-0.208$}& \multicolumn{1}{c}{$-$} & \multicolumn{1}{c}{0.293}\\[1.5em]

    \midrule
\multicolumn{1}{c}{$\alpha=1,~\beta=4$} & \multicolumn{1}{r}{$\omega_{0}$} & \multicolumn{1}{r}{$\omega_{1}$}& 
  \multicolumn{1}{c}{sign}&
  \multicolumn{1}{c}{$|$quot$|$}\\  
    \midrule
\multicolumn{1}{c}{$\gamma$=2} & \multicolumn{1}{c}{0.795} & \multicolumn{1}{r}{0.167}& \multicolumn{1}{c}{+} & \multicolumn{1}{c}{0.210}\\

\multicolumn{1}{c}{$\gamma$=3} & \multicolumn{1}{c}{0.795} & \multicolumn{1}{r}{$-0.125$}& \multicolumn{1}{c}{$-$} & \multicolumn{1}{c}{0.157}\\

\multicolumn{1}{c}{$\gamma$=5} & \multicolumn{1}{c}{0.795} & \multicolumn{1}{r}{$-0.083$}& \multicolumn{1}{c}{$-$} & \multicolumn{1}{c}{0.105}\\[1.5em]

    \midrule
\multicolumn{1}{c}{$\alpha=1,~\beta=5$} & \multicolumn{1}{r}{$\omega_{0}$} & \multicolumn{1}{r}{$\omega_{1}$}& 
  \multicolumn{1}{c}{sign}&
  \multicolumn{1}{c}{$|$quot$|$}\\  
    \midrule   
\multicolumn{1}{c}{$\gamma$=2} & \multicolumn{1}{c}{0.754} & \multicolumn{1}{r}{$-0.208$}& \multicolumn{1}{c}{$-$} & \multicolumn{1}{c}{0.276}\\

\multicolumn{1}{c}{$\gamma$=3} & \multicolumn{1}{c}{0.754} & \multicolumn{1}{r}{$-0.167$}& \multicolumn{1}{c}{$-$} & \multicolumn{1}{c}{0.221}\\

\multicolumn{1}{c}{$\gamma$=4} & \multicolumn{1}{c}{0.754} & \multicolumn{1}{r}{$-0.125$}& \multicolumn{1}{c}{$-$} & \multicolumn{1}{c}{0.166}\\[1.5em]

  \midrule

       \multicolumn{1}{c}{$\alpha=2,~\beta=3$} & \multicolumn{1}{r}{$\omega_{0}$} & \multicolumn{1}{r}{$\omega_{1}$}& 
  \multicolumn{1}{c}{sign}&
  \multicolumn{1}{c}{$|$quot$|$}\\  
    \midrule
\multicolumn{1}{c}{$\gamma$=1} & \multicolumn{1}{c}{0.337} & \multicolumn{1}{r}{$-0.292$}& \multicolumn{1}{c}{$-$} & \multicolumn{1}{c}{0.865}\\

\multicolumn{1}{c}{$\gamma$=4} & \multicolumn{1}{c}{0.337} & \multicolumn{1}{r}{0.042}& \multicolumn{1}{c}{+} & \multicolumn{1}{c}{0.124}\\

\multicolumn{1}{c}{$\gamma$=5} & \multicolumn{1}{c}{0.337} & \multicolumn{1}{r}{0.125}& \multicolumn{1}{c}{+} & \multicolumn{1}{c}{0.371}\\[1.5em]

    \midrule
\multicolumn{1}{c}{$\alpha=2,~\beta=4$} & \multicolumn{1}{r}{$\omega_{0}$} & \multicolumn{1}{r}{$\omega_{1}$}& 
  \multicolumn{1}{c}{sign}&
  \multicolumn{1}{c}{$|$quot$|$}\\  
    \midrule
\multicolumn{1}{c}{$\gamma$=1} & \multicolumn{1}{c}{0.163} & \multicolumn{1}{r}{0.792}& \multicolumn{1}{c}{+} & \multicolumn{1}{c}{4.860}\\

\multicolumn{1}{c}{$\gamma$=3} & \multicolumn{1}{c}{0.163} & \multicolumn{1}{r}{0.458}& \multicolumn{1}{c}{+} & \multicolumn{1}{c}{2.814}\\

\multicolumn{1}{c}{$\gamma$=5} & \multicolumn{1}{c}{0.163} & \multicolumn{1}{r}{0.417}& \multicolumn{1}{c}{+} & \multicolumn{1}{c}{2.558}\\[1.5em]

    \midrule
\multicolumn{1}{c}{$\alpha=2,~\beta=5$} & \multicolumn{1}{r}{$\omega_{0}$} & \multicolumn{1}{r}{$\omega_{1}$}& 
  \multicolumn{1}{c}{sign}&
  \multicolumn{1}{c}{$|$quot$|$}\\  
    \midrule    
\multicolumn{1}{c}{$\gamma$=1} & \multicolumn{1}{c}{0.080} & \multicolumn{1}{r}{1.042}& \multicolumn{1}{c}{+} & \multicolumn{1}{c}{13.095}\\

\multicolumn{1}{c}{$\gamma$=3} & \multicolumn{1}{c}{0.080} & \multicolumn{1}{r}{0.292}& \multicolumn{1}{c}{+} & \multicolumn{1}{c}{3.667}\\

\multicolumn{1}{c}{$\gamma$=4} & \multicolumn{1}{c}{0.080} & \multicolumn{1}{r}{0.333}& \multicolumn{1}{c}{+} & \multicolumn{1}{c}{4.190}\\

    \bottomrule

  \end{tabular*}
\end{table}

\begin{table}[hp]
 \caption{Parallel squares in $[0,1]^5$ for \fec.}
  \label{tab:squares-fec}

  \begin{tabular*}{\linewidth}{@{\extracolsep{\fill}}lllll@{}}
    \toprule
    \multicolumn{1}{c}{$\alpha=1,~\beta=3$} & \multicolumn{1}{r}{$\omega_{0}$} & \multicolumn{1}{r}{$\omega_{1}$}& 
    \multicolumn{1}{c}{sign}&
    \multicolumn{1}{c}{$|$quot$|$}\\  
    \midrule
\multicolumn{1}{c}{$\gamma$=2} & \multicolumn{1}{c}{1.152} & \multicolumn{1}{r}{$-0.577$}& \multicolumn{1}{c}{$-$} & \multicolumn{1}{c}{0.501}\\

\multicolumn{1}{c}{$\gamma$=4} & \multicolumn{1}{c}{1.152} & \multicolumn{1}{r}{1.008}& \multicolumn{1}{c}{+} & \multicolumn{1}{c}{0.876}\\

\multicolumn{1}{c}{$\gamma$=5} & \multicolumn{1}{c}{1.152} & \multicolumn{1}{r}{$-0.336$}& \multicolumn{1}{c}{$-$} & \multicolumn{1}{c}{0.292}\\[1.5em]
  
    \midrule
\multicolumn{1}{c}{$\alpha=1,~\beta=4$} & \multicolumn{1}{r}{$\omega_{0}$} & \multicolumn{1}{r}{$\omega_{1}$}& 
  \multicolumn{1}{c}{sign}&
  \multicolumn{1}{c}{$|$quot$|$}\\  
  \midrule
\multicolumn{1}{c}{$\gamma$=2} & \multicolumn{1}{c}{0.172} & \multicolumn{1}{r}{0.823}& \multicolumn{1}{c}{+} & \multicolumn{1}{c}{4.784}\\

\multicolumn{1}{c}{$\gamma$=3} & \multicolumn{1}{c}{0.172} & \multicolumn{1}{r}{0.029}& \multicolumn{1}{c}{+} & \multicolumn{1}{c}{0.167}\\

\multicolumn{1}{c}{$\gamma$=5} & \multicolumn{1}{c}{0.172} & \multicolumn{1}{r}{$-0.099$}& \multicolumn{1}{c}{$-$} & \multicolumn{1}{c}{0.578}\\[1.5em]

    \midrule
\multicolumn{1}{c}{$\alpha=1,~\beta=5$} & \multicolumn{1}{r}{$\omega_{0}$} & \multicolumn{1}{r}{$\omega_{1}$}& 
  \multicolumn{1}{c}{sign}&
  \multicolumn{1}{c}{$|$quot$|$}\\  
    \midrule    
\multicolumn{1}{c}{$\gamma$=2} & \multicolumn{1}{c}{0.508} & \multicolumn{1}{r}{$-1.138$}& \multicolumn{1}{c}{$-$} & \multicolumn{1}{c}{2.241}\\

\multicolumn{1}{c}{$\gamma$=3} & \multicolumn{1}{c}{0.508} & \multicolumn{1}{r}{$-0.980$}& \multicolumn{1}{c}{$-$} & \multicolumn{1}{c}{1.929}\\

\multicolumn{1}{c}{$\gamma$=4} & \multicolumn{1}{c}{0.508} & \multicolumn{1}{r}{0.236}& \multicolumn{1}{c}{+} & \multicolumn{1}{c}{0.465}\\[1.5em]
    
    \midrule
       \multicolumn{1}{c}{$\alpha=2,~\beta=3$} & \multicolumn{1}{r}{$\omega_{0}$} & \multicolumn{1}{r}{$\omega_{1}$}& 
  \multicolumn{1}{c}{sign}&
  \multicolumn{1}{c}{$|$quot$|$}\\  
  
   \midrule
\multicolumn{1}{c}{$\gamma$=1} & \multicolumn{1}{c}{0.547} & \multicolumn{1}{r}{$-1.183$}& \multicolumn{1}{c}{$-$} & \multicolumn{1}{c}{2.163}\\

\multicolumn{1}{c}{$\gamma$=4} & \multicolumn{1}{c}{0.547} & \multicolumn{1}{r}{1.701}& \multicolumn{1}{c}{+} & \multicolumn{1}{c}{3.111}\\

\multicolumn{1}{c}{$\gamma$=5} & \multicolumn{1}{c}{0.547} & \multicolumn{1}{r}{0.146}& \multicolumn{1}{c}{+} & \multicolumn{1}{c}{0.267}\\[1.5em]

    \midrule
\multicolumn{1}{c}{$\alpha=2,~\beta=4$} & \multicolumn{1}{r}{$\omega_{0}$} & \multicolumn{1}{r}{$\omega_{1}$}& 
  \multicolumn{1}{c}{sign}&
  \multicolumn{1}{c}{$|$quot$|$}\\  
    \midrule
\multicolumn{1}{c}{$\gamma$=1} & \multicolumn{1}{c}{0.577} & \multicolumn{1}{r}{$-0.074$}& \multicolumn{1}{c}{$-$} & \multicolumn{1}{c}{0.128}\\

\multicolumn{1}{c}{$\gamma$=3} & \multicolumn{1}{c}{0.577} & \multicolumn{1}{r}{$-0.577$}& \multicolumn{1}{c}{$-$} & \multicolumn{1}{c}{1.000}\\

\multicolumn{1}{c}{$\gamma$=5} & \multicolumn{1}{c}{0.577} & \multicolumn{1}{r}{0.065}& \multicolumn{1}{c}{+} & \multicolumn{1}{c}{0.113}\\[1.5em]

    \midrule
\multicolumn{1}{c}{$\alpha=2,~\beta=5$} & \multicolumn{1}{r}{$\omega_{0}$} & \multicolumn{1}{r}{$\omega_{1}$}& 
  \multicolumn{1}{c}{sign}&
  \multicolumn{1}{c}{$|$quot$|$}\\  
	\midrule 
\multicolumn{1}{c}{$\gamma$=1} & \multicolumn{1}{c}{0.739} & \multicolumn{1}{r}{$-0.908$}& \multicolumn{1}{c}{$-$} & \multicolumn{1}{c}{1.229}\\

\multicolumn{1}{c}{$\gamma$=3} & \multicolumn{1}{c}{0.739} & \multicolumn{1}{r}{0.338}& \multicolumn{1}{c}{+} & \multicolumn{1}{c}{0.458}\\

\multicolumn{1}{c}{$\gamma$=4} & \multicolumn{1}{c}{0.739} & \multicolumn{1}{r}{1.250}& \multicolumn{1}{c}{+} & \multicolumn{1}{c}{1.693}\\

    \bottomrule

\end{tabular*}
\end{table}

\begin{table}
  \caption{Numbering of the genotypes.}
  \label{tab:genotypes}
  \begin{tabular*}{.95\linewidth}{@{\extracolsep{\fill}}rll@{}}
    \toprule
    \multicolumn{1}{c}{number} & \multicolumn{1}{c}{genotype} & \multicolumn{1}{c}{comment} \\
    \midrule
    0 & 0 0 0 0 0  & germ-free\\
    1 & 1 0 0 0 0  & \emph{Lactobacillus plantarum}\\
    2 & 0 1 0 0 0  & \emph{Lactobacillus brevis}\\
    3 & 0 0 1 0 0  & \emph{Acetobacter pasteurianus} \\
    4 & 0 0 0 1 0  & \emph{Acetobacter tropicalis} \\
    5 & 0 0 0 0 1  & \emph{Acetobacter orientalis}\\
    6 & 1 1 0 0 0  & bacterial growth decreases: competition\\
    7 & 1 0 1 0 0  & bacterial growth increases: synergy \\
    8 & 1 0 0 1 0  & synergy\\
    9 & 1 0 0 0 1  & synergy\\
    10 & 0 1 1 0 0  & synergy\\
    11 & 0 1 0 1 0  & synergy\\
    12 & 0 1 0 0 1  & synergy\\
    13 & 0 0 1 1 0  & competition \\
    14 & 0 0 1 0 1  & competition\\
    15 & 0 0 0 1 1  & competition\\
    16 & 1 1 1 0 0  & \\
    17 & 1 1 0 1 0  & \\
    18 & 1 1 0 0 1  & \\
    19 & 1 0 1 1 0  & \\
    20 & 1 0 1 0 1  & \\
    21 & 1 0 0 1 1  & \\
    22 & 0 1 1 1 0  & \\
    23 & 0 1 1 0 1  & \\
    24 & 0 1 0 1 1  & \\
    25 & 0 0 1 1 1  & \\
    26 & 1 1 1 1 0  & \\
    27 & 1 1 1 0 1  & \\
    28 & 1 1 0 1 1  & \\
    29 & 1 0 1 1 1  & \\
    30 & 0 1 1 1 1  & \\
    31 & 1 1 1 1 1 \\
    \bottomrule
  \end{tabular*}
\end{table}

\input{paralleltransport_facet2}

\begin{figure}[th]
\includegraphics[width=\textwidth]{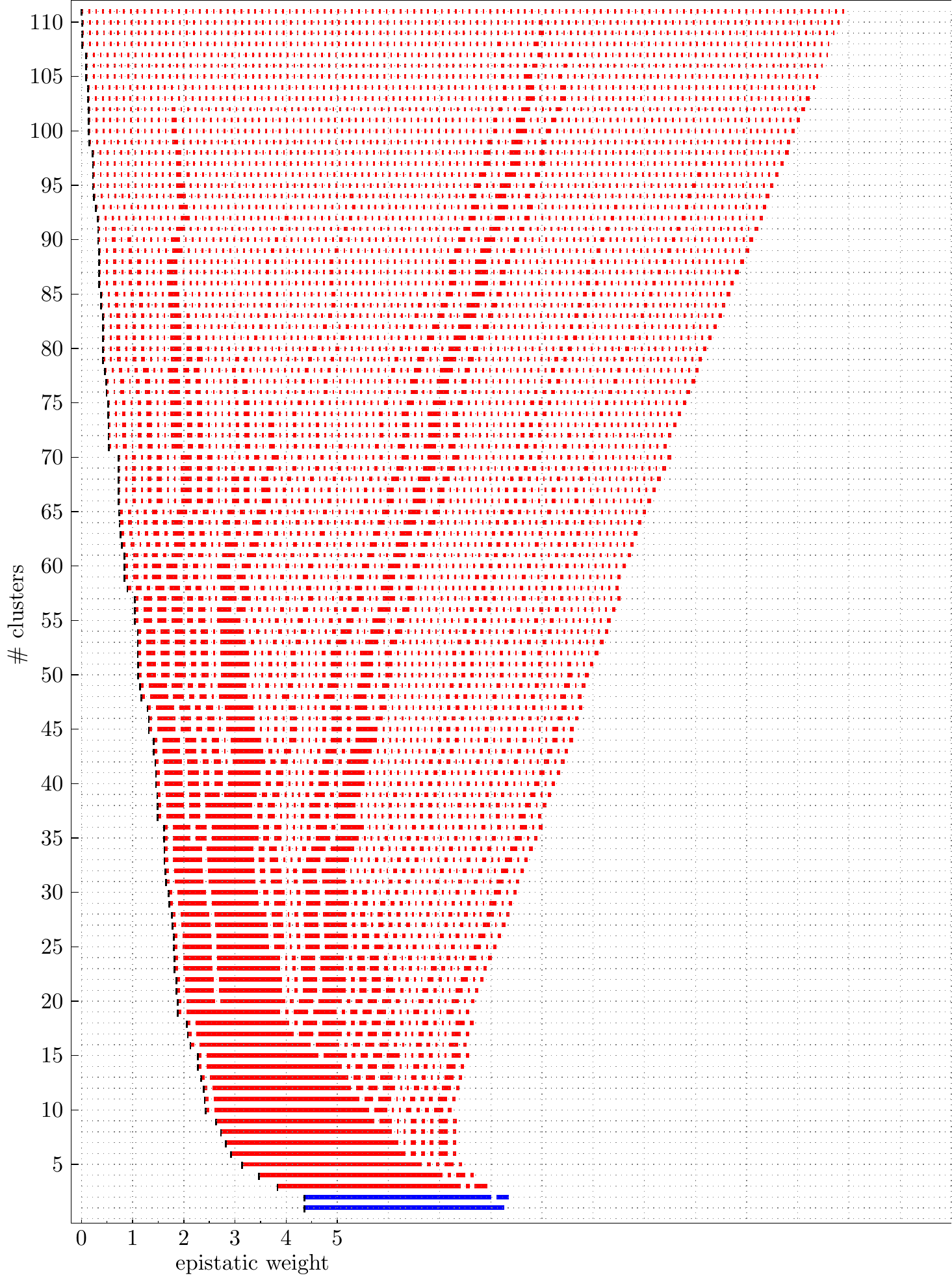}
\caption{Epistatic filtration of the triangulation $\subdivision(\{0,1\}^5,\text{ttd})$ with 111 maximal cells. The colors of the bars represent various levels of statistical significance 
 ($p<0.05$: \textcolor{blue}{blue}, $0.05\leq p<0.1$: \textcolor{purple}{purple}, $p\geq 0.1$: \textcolor{red}{red}); cf.\ Section~\ref{sec:statsignificance1}. 
  The $f$-vector for $\subdivision(\{0,1\}^5,\text{ttd})$ reads (32, 204, 540, 702, 446, 111)  and (111, 220, 137, 28, 1) for its tight span.}
\label{fig:whole_ttd_filtration}
\end{figure}






\end{document}